\documentclass[submission,copyright,creativecommons]{eptcs}
\providecommand{\event}{WWV 2014} % Name of the event you are submitting to
\usepackage{breakurl}             % Not needed if you use pdflatex only.
\usepackage{amsbsy} 
\usepackage{amsfonts}
\usepackage{graphicx}
\usepackage{dcolumn}
\usepackage{amsmath}
\usepackage{latexsym}
\usepackage{amssymb}
\usepackage{epsfig}
\usepackage{verbatim} 
\usepackage{epsfig,graphics}
\usepackage{amsmath, amsthm, amssymb}
\usepackage{subfigure}
\usepackage{tikz}
\usetikzlibrary{positioning}

\title{A Local Logic for Realizability in Web Service Choreographies}
\author{R. Ramanujam
\institute{IMSc, Chennai}
\email{jam@imsc.res.in}
\and
S. Sheerazuddin
\institute{SSNCE, Chennai}
\email{sheerazuddins@ssn.edu.in}
}
\def\titlerunning{A Local Logic for Realizability in Web Service Choreographies}
\def\authorrunning{R. Ramanujam \& S. Sheerazuddin}

\newtheorem{dfn}{Definition}[section]

\newtheorem{lem}[dfn]{Lemma}
\newtheorem{thm}[dfn]{Theorem}

\newcommand{\RightBox}{{\phantom{a}}\hfill $\Box$ \\}
\newcommand{\diamin}{\Diamond\kern-0.5em{\raisebox{.25ex}{\rm -}}\kern0.175em}

\newcommand{\nxt}{\mbox{$\bigcirc$}}

\newcommand{\Imply}{~~\mathbf{\supset}~~}
\newcommand{\Not}{\mbox{$\lnot$}}
\newcommand{\xor}{\mathbf{\oplus}}
\newcommand{\True}{\mathit{True}}
\newcommand{\False}{\mathit{False}}

\newcommand{\cupover}{\displaystyle \bigcup}

\newcommand{\otilde}[1]{\widetilde{#1}}

\newcommand{\Sigtil}{\mbox{$\otilde{\Sigma}$}}
\newcommand{\DA}{\mbox{$\Sigtil~=~(\Sigma_1, \dots, \Sigma_n)$}}

\newcommand{\defn}{\mbox{$~\stackrel{\rm def}{=}~$}}

\newcommand{\step}[1]{\mbox{$\stackrel{#1}{\to}$}}
\newcommand{\Funnyto}{\rightsquigarrow}
\newcommand{\longstep}[1]{\mbox{$\stackrel{#1}{\longrightarrow}$}}

\newcommand{\To}{\Rightarrow}
\newcommand{\From}{\Leftarrow}

\newcommand{\presca}[1]{\mbox{${ }^{\bullet}#1$}}
\newcommand{\postsca}[1]{\mbox{$#1 \, { }^{\bullet}$}}%

\newcommand{\ldot}{{\rm <}\kern-0.37em{\raisebox{.25ex}{\bf .}}\kern0.375em}

\newcommand{\calL}{\mathcal{L}}

\newcommand{\posetlang}[1]{\mbox{${\calL}^{po}(#1)$}}%poset language of an SCA
\newcommand{\calC}{\mathcal{C}}

\newcommand{\calM}{\mathcal{M}}

\begin{document}
\maketitle

\begin{abstract}
Web service choreographies specify conditions on observable interactions among the services. An important question in this regard is realizability: given a choreography $C$, does there exist a set of service implementations $I$ that conform to $C$ ? Further, if $C$ is realizable, is there an algorithm to construct implementations in $I$ ? We propose a local temporal logic in which choreographies can be specified, and for specifications in the logic, we solve the realizability problem by constructing service implementations (when they exist) as communicating automata. These are nondeterministic finite state automata with a coupling relation. We also report on an implementation of the realizability algorithm and discuss experimental results. 
\end{abstract}

\section{Introduction}
The study of composition of distributed web services has received great attention.
When we know what kind of services are available,
specifying a sequence of communications to and from them can well
suffice to describe the overall service required. Such a {\em global} 
specification of interaction composition has been termed {\bf choreography}
(\cite{BFHS}).  The distributed services can then be synthesized as 
autonomous agents that interact in conformance with the given choreography. 
This offers an abstract methodology for the design and development of web 
services. The choreography and its implementation may be put together in 
a {\bf choreography model} \cite{SBFZ07}, $M=(C,I)$, where, as already 
mentioned, $C$ is a specification of the desired global behaviours (a 
choreography), and $I$ a representation of local services and their 
local behaviours (an implementation) which collectively should satisfy 
the specified global behaviour. A {\bf choreography modeling language} 
\cite{SBFZ07} provides the means to define choreography models, i.e., 
choreographies, service implementations, and their semantics including 
a mechanism to compare global behaviors generated by service implementations 
with a choreography. 

A principal challenge for such a methodology is that choreographies
be {\bf realizable} (\cite{wscdl}). What may seem simple global specifications may
yet be hard, or even impossible, to implement as a composition of
distributed services. The reason is simple: while the global specification
requires a communication between $1$ and $2$ to precede that between
$3$ and $4$, the latter, lacking knowledge of the former, may well
communicate earlier. Thus the composition would admit forbidden
behaviours. In general, many seemingly innocuous specifications
may be unrealizable. Even checking whether a choreography is realizable
may be hard, depending on the expressive power of the formalism in
which the choreography is specified.

Closely related, but more manageable, is the problem of {\bf conformance}:
check whether a given set of services implement the given choreography
specification. Once again, the expressive power of the specification
formalism is critical for providing algorithmic solutions to the problem.

The two problems relate to the satisfiability and model checking problems
of associated logics. Since the 1980's a rich body of literature has been
built in the study of such problems (\cite{CGP}).

In the literature, choreographies have been formally specified using
automata \cite{FBS}, UML collaboration diagrams \cite{BF08}, interaction
Petri nets \cite{DW07}, or process algebra \cite{CHY07}. The service
implementations have been modelled variously as Mealy machines, Petri
nets or process algebra. Visual formalisms (such as message sequence
charts \cite{RGG96}) are naturally attractive and intuitive for 
choreography specifications but can be imprecise. For instance,
it is hard to distinguish between interactions that are permissible and
those that must indeed take place. While machine models are precise they
might require too much detail. 

A natural idea in this context is the use of a logical formalism for 
choreography specification and that of finite state machines for their 
implementation. When the formulas of the logic specify global interaction
behaviour and models for the logic are defined using products of machines, 
realizability and conformance naturally correspond to the satisfiability
and model checking problems for the logic.

Once again, a natural candidate for such a logic is that of temporal
logic, linear time or branching time (\cite{P77}, \cite{CGP}). One
difficulty with the use of temporal logics for global interaction
specifications is that sequentiality is natural in such logics but
concurrency poses challenges. It is rather easy to come up with
specifications in temporal logics that are not realizable. On the
other hand, if we wish to algorithmically decide whether a temporal
specification is realizable or not, the problem is often undecidable,
and in some cases of high complexity even when decidable. (See 
\cite{AEY05} and \cite{AMNN05} for decidability of the closely related
problem of realizability of message sequence graphs.)

One simple way out is to design the temporal logic, limit its 
expressiveness drastically, so that we can ensure {\em by diktat}
that every satisfiable formula in it is realizable. This is the line we follow
here, initiated by \cite{T95} and developed by \cite{R96}, \cite{MR}.
In such {\bf local} temporal logics, we can ensure realizability by 
design. The models for these logics are presented as a system
of communicating automata (SCA). Both realizability and conformance
are decidable in this setting.

Our work is similar to that of McNeile \cite{McN10} who
extend the process algebra based formalism of Protocol Modeling
\cite{MS06} to define a notion of protocol contract and describe
choreographies and participant contracts. They give sufficient conditions
for realizability in both synchronous and asynchronous collaborations.

The language-based choreography realizability problem considered
in this paper was proposed for conversation protocols in \cite{FBS}
where sufficient conditions for realizability were given. Halle \&
Bultan \cite{HB10} consider the realizability of a particular class of
choreographies called arbitrary-initiator protocols for which sufficiency
conditions given in \cite{FBS} fail. The algorithm for choreography
realizability works by computing a finite-state model that keeps track
of the information about the global state of a conversation protocol that
each peer can deduce from the messages it sends and receives. Thereafter,
the realizability can be checked by searching for disagreements between
peers' deduced states.

\cite{SBR12} model choreography as UML collaboration diagrams
and check their realizability. They have also
implemented a tool which not only checks the realizability of choreography
specified using collaboration diagrams but also synthesize the service
implementations that realize the choreography \cite{BFF09}.

\cite{BBO12} consider the realizability problem for
choreographies modelled as conversation protocols \cite{FBS}(finite
automata over send events). They give necessary and sufficient
conditions which need to be satisfied by the conversations for them to be
realizable. They implement the proposed realizability check and show that
it can efficiently determine the realizability of  a subclass of
contracts \cite{Fahn06} and UML collaboration diagrams \cite{BF08},
apart from conversation protocols.

The work on session types \cite{HYC08} is also related to realizability of
conversation protocols and has been used as a formal basis for modelling
choreography languages \cite{CHY07}.

The work presented in \cite{LW09} checks choreography realizability
using the concept of controllability. Given a choreography description, a
monitor service is computed from that choreography. The monitor service
is used as a centralized orchestrator of the interaction to compute
the distributes peers. The choreography is said to be realizable if
the monitor service is controllable, that is, there exists a set of
peers such that the composition of the monitor service and those peers
is deadlock-free.

Pistore et al., \cite{KP06} present a formal framework for the definition
of both global choreography as well as local peer implementations. They
introduce a hierarchy of realizability notions that allows for capturing
various properties of the global specifications, and associate specific
communication models to each of them. Finally, they present an approach,
based on the analysis of communication models, to associate a particular
level of realizability to the choreography.

Decker and Weske \cite{DW07} model choreography as interaction Petri
nets, an extension of Petri nets for interaction modelling, and propose
an algorithm for deriving corresponding behavioural interfaces. The
message exchanges are assumed to be asynchronous in nature. They define
two properties, realizability and local enforceability, for interaction
Petri nets and introduce algorithms for checking these properties.

In the light of this literature, the need and relevance of the work
presented in this paper is naturally questionable. The contribution
of this paper is two fold: one, to argue that {\bf partial orders}
provide a natural way of describing potential interactions -- sequential
and concurrent; two, to suggest that a test for realizability be
translated to an expressiveness restriction on the formalism for
choreography specification in such a way that {\em every} consistent
specification is realizable. While realizability by design is not in itself 
new, the application of automata based methods on partial orders can lead to
new ways of defining `good' choreographies. A major advantage of such partial order
based specification is that we can reason about components (services)
separately, and limit global reasoning to the minimum required. In 
terms of worst case complexity this makes no difference, but in practice 
this is of great use.

The paper is organized as follows. In the next section, we define
choreography realizability and discuss examples of realizable
and unrealizable choreographies. Further, we note that even though
choreographies, modelled as conversation protocols, are defined over
send events they give rise to partial order behaviours in service
implementations.  We then propose $p$-LTL, which admits only realizable 
choreographies, and present systems of communicating automata (SCA) to 
model sets of service implementations admitted by choreographies in 
$p$-LTL. We show that realizability is decidable, and discuss some 
experiments in implementing the decision algorithm. Detailed proofs are 
relegated to the Appendix.

\section{Choreography realizability}
A choreography specification $C$ is {\bf realizable} if there is an 
{\bf implementation} $I$ of the interacting services such that once
the system is initialized, its processes behave according to
the choreography specification.  

\begin{figure*}
\begin{center}
\begin{tikzpicture}

%----------------------------------------------------------------------------------------

	\node 	(f0)		at	(-4,0) 	[rectangle, thick, draw=red!50,fill=green!20,inner sep=0pt,minimum size=0.85cm]{$a$};

	\node 	(f1)		at	(0,0) 	[rectangle, thick, draw=red!50,fill=green!20,inner sep=0pt,minimum size=0.85cm]{$j$};

	\node 	(f2)		at	(4,0) 	[rectangle, thick, draw=red!50,fill=green!20,inner sep=0pt,minimum size=0.85cm]{$h$};

	\node 	(f21)		at	(3.7,0.4) 	[rectangle]{};
	\node 	(f22)		at	(3.7,-0.4) 	[rectangle]{};

	\node 	(f11)		at	(0.35,0.4) 	[rectangle]{};
	\node 	(f12)		at	(0.35,-0.4) 	[rectangle]{};

	\node 	(f13)		at	(-0.3,0.4) 	[rectangle]{};
	\node 	(f14)		at	(-0.3,-0.4) 	[rectangle]{};

	\node 	(f01)		at	(-3.7,0.4) 	[rectangle]{};
	\node 	(f02)		at	(-3.7,-0.4) 	[rectangle]{};

	\draw[<-, thick]	(f01)	to	node[auto]{$query$}	(f13);
	\draw[->, thick]	(f02)	to	node[auto,swap]{$suggest$}	(f14);

	\draw[->, thick]	(f12)	to	node[auto,swap]{$reserve$}	(f22);
	\draw[->, thick]	(f21)	to	node[auto,swap]{$confirm$}	(f11);

%-------------------------------------------------------------------------------------------------------
	\node 	(g)		at	(-5,3) 	[circle]{};

	\node 	(g0)		at	(-4,3) 	[circle, thick, draw=red!50,fill=blue!20,inner sep=0pt,minimum size=0.40cm]{};

	\node 	(g1)		at	(-2,3) 	[circle, thick, draw=red!50,fill=blue!20,inner sep=0pt,minimum size=0.40cm]{};

	\node 	(g2)		at	(0,3) 	[circle, thick, draw=red!50,fill=blue!20,inner sep=0pt,minimum size=0.40cm]{};

	\node 	(g3)		at	(2,3) 	[circle, thick, draw=red!50,fill=blue!20,inner sep=0pt,minimum size=0.40cm]{};

	\node 			at	(4,3) 	[circle, thick, draw=black!50,inner sep=0pt,minimum size=0.50cm]{};

	\node 	(g4)		at	(4,3) 	[circle, thick, draw=red!50,fill=blue!20,inner sep=0pt,minimum size=0.40cm]{};

	\node 	(g5)		at	(0,5) 	[circle, thick, draw=red!50,fill=blue!20,inner sep=0pt,minimum size=0.40cm]{};

	\node 			at	(5,3) 	[circle,]{$:C_0$};

	\draw[->, very thick]	(g)	to	node[auto,swap]{}	(g0);

	\draw[->, thick]	(g0)	to	node[auto,swap]{$query$}	(g1);

	\draw[->, thick]	(g1)	to	node[auto,swap]{$suggest$}	(g2);

	\draw[->, thick]	(g2)	to	node[auto,swap]{$reserve$}	(g3);

	\draw[->, thick]	(g3)	to	node[auto,swap]{$confirm$}	(g4);

	\draw[->, bend left=30, thick]	(g2)	to	node[auto]{$query$}	(g5);

	\draw[->, bend left=30, thick]	(g5)	to	node[auto]{$suggest$}	(g2);

%----------------------------------------------------------------------------------------------------------------------------

\end{tikzpicture}
\end{center}
\caption{A Realizable Choreography}
\label{fig:real-chor}
\end{figure*}
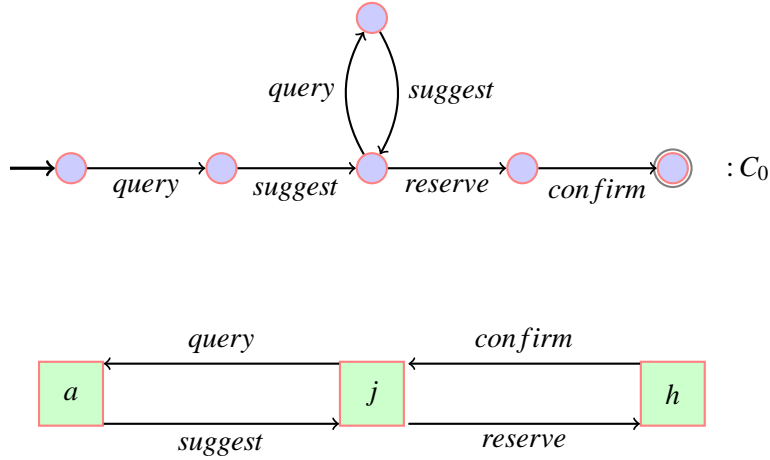 

Consider the choreography $C_0$ \cite{BFS07} represented as conversation
protocol, given in Figure \ref{fig:real-chor}. The conversation
protocol, a collection of sequences  of send events (conversations), is modelled as 
a nondeterministic finite state automaton.  Let $L(C_0)$ denote the set of all conversations
in $C_0$.  In $C_0$ there are three services interacting with each other:
John ($j$), Agent ($a$) and Hotel ($h$). John wants to take a vacation. He has certain constraints
about where he wants to vacation, so he sends a query to his Agent stating his
constraints and asking for advice. The Agent responds to John’s query by sending him a suggestion. If John is not happy with the Agent’s suggestion he sends another query requesting another suggestion. Eventually, John makes up his
mind and sends a reservation request to the hotel he picks. The hotel responds
to John’s reservation request with a confirmation message.

One set of service implementations $I_0$ for the choreography $C_0$
is given in Figure \ref{fig:real-impl}. Each service is implemented
as a nondeterministic finite state automaton (NFA) over send and
receive events, equipped with FIFO queues for sending message to other services.
  The causal dependence among these events (e.g. that
a message can be received only after it is sent) is represented 
by couplings, shown as $\lambda$ labelled arrows. For
example, consider the case when service $s_2$ sends a message $query$
to service $s_1$. The transition $(q_0,?query,q_1)$ in $I_0$ must be
coupled with $(q_4,!query,q_5)$, that is, $(q_0,?query,q_1)$ must happen
before  $(q_4,!query,q_5)$. This coupling is represented by a $\lambda$
labelled arrow from the target state of $!query$ event ($q_1$) to target
state of $?query$ event ($q_5$).

\begin{figure*}
\begin{center}
\begin{tikzpicture}[scale=0.6]

%----------------------------------------------------------------------------------------

	\node 	(f)		at	(0,8) 	[circle]{$a$};

	\node 	(f0)		at	(0,6) 	[circle, thick, draw=red!60,fill=blue!20,inner sep=0pt,minimum size=0.50cm]{$q_0$};

	\node 	(f1)		at	(0,4) 	[circle, thick, draw=red!60,fill=blue!20,inner sep=0pt,minimum size=0.50cm]{$q_1$};

	\node 	(f2)		at	(0,2) 	[circle, thick, draw=red!60,fill=blue!20,inner sep=0pt,minimum size=0.50cm]{$q_2$};

	\node 			at	(0,2) 	[circle, thick, draw=black!60,inner sep=0pt,minimum size=0.60cm]{};

	\node 	(f3)		at	(-3,2) 	[circle, thick, draw=red!60,fill=blue!20,inner sep=0pt,minimum size=0.50cm]{$q_3$};

	\draw[->, very thick]	(f)	to	node[auto]{}	(f0);

	\draw[->, thick]	(f0)	to	node[auto]{$?query$}	(f1);

	\draw[->, thick]	(f1)	to	node[auto]{$!suggest$}	(f2);

	\draw[->, bend left=30, thick]	(f2)	to	node[auto]{$?query$}	(f3);

	\draw[->, bend left=30, thick]	(f3)	to	node[auto]{$!suggest$}	(f2);

%-------------------------------------------------------------------------------------------------------
	\node 	(g)		at	(6,8) 	[circle]{$j$};

	\node 	(g0)		at	(6,6) 	[circle, thick, draw=red!60,fill=blue!20,inner sep=0pt,minimum size=0.50cm]{$q_4$};

	\node 	(g1)		at	(6,4) 	[circle, thick, draw=red!60,fill=blue!20,inner sep=0pt,minimum size=0.50cm]{$q_5$};

	\node 	(g2)		at	(6,2) 	[circle, thick, draw=red!60,fill=blue!20,inner sep=0pt,minimum size=0.50cm]{$q_6$};

	\node 	(g3)		at	(6,0) 	[circle, thick, draw=red!60,fill=blue!20,inner sep=0pt,minimum size=0.50cm]{$q_7$};

	\node 			at	(6,-2) 	[circle, thick, draw=black!60,inner sep=0pt,minimum size=0.60cm]{};

	\node 	(g4)		at	(6,-2) 	[circle, thick, draw=red!60,fill=blue!20,inner sep=0pt,minimum size=0.50cm]{$q_8$};

	\node 	(g5)		at	(3,2) 	[circle, thick, draw=red!60,fill=blue!20,inner sep=0pt,minimum size=0.50cm]{$q_9$};

	\draw[->, very thick]	(g)	to	node[auto]{}	(g0);

	\draw[->, thick]	(g0)	to	node[auto]{$!query$}	(g1);

	\draw[->, thick]	(g1)	to	node[auto]{$?suggest$}	(g2);

	\draw[->, thick]	(g2)	to	node[auto]{$!reserve$}	(g3);

	\draw[->, thick]	(g3)	to	node[auto]{$?confirm$}	(g4);

	\draw[->, bend left=30, thick]	(g2)	to	node[auto]{$!query$}	(g5);

	\draw[->, bend left=30, thick]	(g5)	to	node[auto]{$?suggest$}	(g2);

%----------------------------------------------------------------------------------------

	\node 	(h)		at	(11,8) 	[circle]{$h$};

	\node 	(h0)		at	(11,6) 	[circle, thick, draw=red!60,fill=blue!20,inner sep=0pt,minimum size=0.50cm]{};

	\node 	(h1)		at	(11,4) 	[circle, thick, draw=red!60,fill=blue!20,inner sep=0pt,minimum size=0.50cm]{};

	\node 	(h2)		at	(11,2) 	[circle, thick, draw=red!60,fill=blue!20,inner sep=0pt,minimum size=0.50cm]{};

	\node 			at	(11,2) 	[circle, thick, draw=black!60,inner sep=0pt,minimum size=0.60cm]{};

	\draw[->, very thick]	(h)	to	node[auto]{}	(h0);

	\draw[->, thick]	(h0)	to	node[auto]{$?reserve$}	(h1);

	\draw[->, thick]	(h1)	to	node[auto]{$!confirm$}	(h2);

%--------------------------------------------------------------------------------------------------------------
%-------------couplings----------------------------------------------------------------------------------------

	\draw[->,very thin]	(g1)	to	node[auto,swap]{$\lambda$}	(f1);

	\draw[->,bend right=70,very thin]	(f2)	to	node[auto,swap]{$\lambda$}	(g2);

	\draw[<-,bend right=70,very thin]	(f3)	to	node[auto,swap]{$\lambda$}	(g5);

	\draw[->,bend right=30,very thin]	(g3)	to	node[auto]{$\lambda$}	(h1);

	\draw[<-,bend right=30,very thin]	(g4)	to	node[auto,swap]{$\lambda$}	(h2);

\end{tikzpicture}
\end{center}
\caption{A Sample Service Implementation for Choreography in Figure \ref{fig:real-chor}}
\label{fig:real-impl}
\end{figure*}
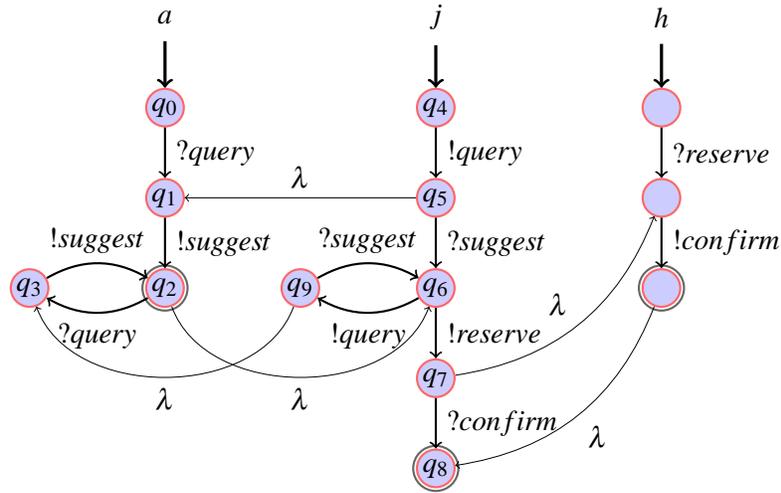

Considering the behaviour specified by $C_0$ as a collection of sequences
of send events is natural and simple, but hides concurrency information.
Two send events by distinct services, locally determined by them, can
proceed simultaneously or in any order. Therefore, the interactions among
service implementations $I_0$ are better viewed as partial orders on
send and receive events.  Let $\calL$ be the set of all such partially 
ordered executions of $I_0$. We call the objects in $\calL$ as {\bf diagrams}. 
A sample of the diagrams in $\calL$ is given in Figure \ref{fig:po-excn}.

Let $D$ be an arbitrary diagram in $\calL$ and $Lin(D)$ be the set of all
linearizations of $D$. A linearization of $D$ is a linear order on the events 
in $D$ which respects the given partial order.
Let $Chor(D)$ be obtained from $Lin(D)$ by removing
receive events from each sequence. Let $Chor(\calL)=\cupover_{D \in
\calL}Chor(D)$. We say that $I_0$ realizes $C_0$ if $Chor(\calL)=L(C_0)$.
For our example, this is easily seen to be the case.

\input{fig5}

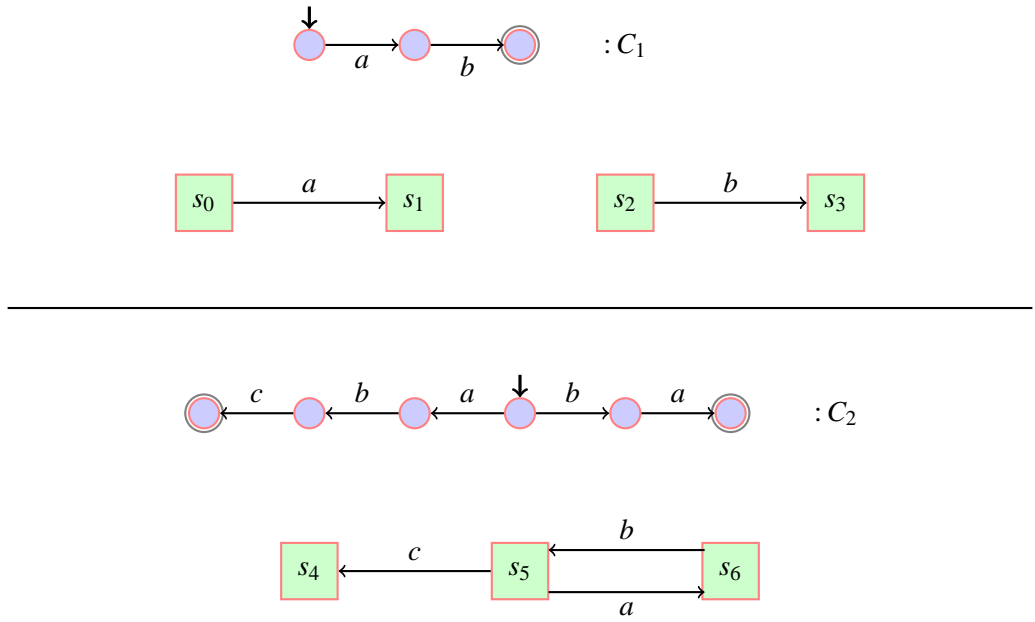
\begin{figure*}
\begin{center}
\begin{tikzpicture}[scale=0.7]

%----------------------------------------------------------------------------------------

	\node 	(f0)		at	(-4,0) 	[rectangle, thick, draw=red!50,fill=green!20,inner sep=0pt,minimum size=0.75cm]{$s_4$};

	\node 	(f1)		at	(0,0) 	[rectangle, thick, draw=red!50,fill=green!20,inner sep=0pt,minimum size=0.75cm]{$s_5$};

	\node 	(f2)		at	(4,0) 	[rectangle, thick, draw=red!50,fill=green!20,inner sep=0pt,minimum size=0.75cm]{$s_6$};

	\node 	(f2')		at	(3.7,0.4) 	[rectangle]{};
	\node 	(f2'')		at	(3.7,-0.4) 	[rectangle]{};

	\node 	(f1')		at	(0.35,0.4) 	[rectangle]{};
	\node 	(f1'')		at	(0.35,-0.4) 	[rectangle]{};

	\draw[->, thick]	(f1)	to	node[auto,swap]{$c$}	(f0);
	\draw[->, thick]	(f1'')	to	node[auto,swap]{$a$}	(f2'');
	\draw[->, thick]	(f2')	to	node[auto,swap]{$b$}	(f1');

%-------------------------------------------------------------------------------------------------------
	\node 	(g)		at	(0,4) 	[circle]{};

	\node 	(g0)		at	(0,3) 	[circle, thick, draw=red!50,fill=blue!20,inner sep=0pt,minimum size=0.40cm]{};

	\node 	(g1)		at	(-2,3) 	[circle, thick, draw=red!50,fill=blue!20,inner sep=0pt,minimum size=0.40cm]{};

	\node 	(g2)		at	(-4,3) 	[circle, thick, draw=red!50,fill=blue!20,inner sep=0pt,minimum size=0.40cm]{};

	\node 	(g3)		at	(-6,3) 	[circle, thick, draw=red!50,fill=blue!20,inner sep=0pt,minimum size=0.40cm]{};

	\node 			at	(-6,3) 	[circle, thick, draw=black!50,inner sep=0pt,minimum size=0.50cm]{};

	\node 	(g1')		at	(2,3) 	[circle, thick, draw=red!50,fill=blue!20,inner sep=0pt,minimum size=0.40cm]{};

	\node 	(g2')		at	(4,3) 	[circle, thick, draw=red!50,fill=blue!20,inner sep=0pt,minimum size=0.40cm]{};

	\node 			at	(4,3) 	[circle, thick, draw=black!50,inner sep=0pt,minimum size=0.50cm]{};

	\node 			at	(6,3) 	[circle,]{$:C_2$};

	\draw[->, very thick]	(g)	to	node[auto,swap]{}	(g0);

	\draw[->, thick]	(g0)	to	node[auto,swap]{$a$}	(g1);

	\draw[->, thick]	(g1)	to	node[auto,swap]{$b$}	(g2);

	\draw[->, thick]	(g2)	to	node[auto,swap]{$c$}	(g3);

	\draw[->, thick]	(g0)	to	node[auto]{$b$}	(g1');

	\draw[->, thick]	(g1')	to	node[auto]{$a$}	(g2');

	\node 	(h')		at	(10,5) 		[circle,]{};

	\node 	(h)		at	(-10,5) 	[circle,]{};

	\draw[-,thick]	(h)	to	node[auto]{}	(h');

%----------------------------------------------------------------------------------------------------------------------------

	\node 	(f3)		at	(-6,7) 	[rectangle, thick, draw=red!50,fill=green!20,inner sep=0pt,minimum size=0.75cm]{$s_0$};

	\node 	(f4)		at	(-2,7) 	[rectangle, thick, draw=red!50,fill=green!20,inner sep=0pt,minimum size=0.75cm]{$s_1$};

	\node 	(f5)		at	(2,7) 	[rectangle, thick, draw=red!50,fill=green!20,inner sep=0pt,minimum size=0.75cm]{$s_2$};

	\node 	(f6)		at	(6,7) 	[rectangle, thick, draw=red!50,fill=green!20,inner sep=0pt,minimum size=0.75cm]{$s_3$};

	\draw[->, thick]	(f3)	to	node[auto]{$a$}	(f4);

	\draw[->, thick]	(f5)	to	node[auto]{$b$}	(f6);

	\node 	(g')		at	(-4,11) 	[circle]{};

	\node 	(g4)		at	(-4,10) 	[circle, thick, draw=red!50,fill=blue!20,inner sep=0pt,minimum size=0.40cm]{};

	\node 	(g5)		at	(-2,10) 	[circle, thick, draw=red!50,fill=blue!20,inner sep=0pt,minimum size=0.40cm]{};

	\node 	(g6)		at	(0,10) 	[circle, thick, draw=red!50,fill=blue!20,inner sep=0pt,minimum size=0.40cm]{};

	\node 			at	(0,10) 	[circle, thick, draw=black!50,inner sep=0pt,minimum size=0.50cm]{};

	\node 			at	(2,10) 	[circle,]{$:C_1$};

	\draw[->, very thick]	(g')	to	node[auto,swap]{}	(g4);

	\draw[->, thick]	(g4)	to	node[auto,swap]{$a$}	(g5);

	\draw[->, thick]	(g5)	to	node[auto,swap]{$b$}	(g6);

%	\node 	(e) at (2.5,-2) [rectangle, thick]{$\chi(e,j)=(\down e-\down e')\cup (\up e-\up e'')$};

%----------------------------------------------------------------------------------------

%----------------------------------------------------------------------------------------

\end{tikzpicture}
\end{center}
\caption{Unrealizable choreographies \cite{SBFZ07}}
\label{fig:un-chor}
\end{figure*} 

Even though $C_0$ turned out to be realizable, we can easily
fashion choreographies that are not realizable \cite{BFHS}. Consider
the choreography $C_1$ with $L(C_1)=\{a\cdot b\}$, given in Figure
\ref{fig:un-chor}, where service $s_0$ sends a message $a$ to $s_1$ and
$s_2$ sends $b$ to $s_3$ . 

Note that there is no causal dependence between a send event $!a$ 
and a send event $!b$ and hence any set of service implementations
defined by projection will also admit the global behaviour $b\cdot a$.
This situation is illustrated in Figure \ref{fig:un-impl}. 

\begin{figure*}
\begin{center}
\begin{tikzpicture}[scale=0.7]

%----------------------------------------------------------------------------------------------------------------------------

	\node 	(g')		at	(-4,11.5) 	[circle]{$s_0$};

	\node 	(g4)		at	(-4,10) 	[circle, thick, draw=red!50,fill=blue!20,inner sep=0pt,minimum size=0.40cm]{};

	\node 	(g5)		at	(-4,8) 	[circle, thick, draw=red!50,fill=blue!20,inner sep=0pt,minimum size=0.40cm]{};

	\node 			at	(-4,8) 	[circle, thick, draw=black!50,inner sep=0pt,minimum size=0.50cm]{};

	\draw[->, very thick]	(g')	to	node[auto,swap]{}	(g4);

	\draw[->, thick]	(g4)	to	node[auto,swap]{$!a$}	(g5);

	\node 	(h')		at	(-2,11.5) 	[circle]{$s_1$};

	\node 	(h4)		at	(-2,10) 	[circle, thick, draw=red!50,fill=blue!20,inner sep=0pt,minimum size=0.40cm]{};

	\node 	(h5)		at	(-2,8) 	[circle, thick, draw=red!50,fill=blue!20,inner sep=0pt,minimum size=0.40cm]{};

	\node 			at	(-2,8) 	[circle, thick, draw=black!50,inner sep=0pt,minimum size=0.50cm]{};

	\draw[->, very thick]	(h')	to	node[auto,swap]{}	(h4);

	\draw[->, thick]	(h4)	to	node[auto,swap]{$?a$}	(h5);

%----------------------------------------------------------------------------------------
	\draw[->, bend right=30, very thin]	(g5)	to	node[auto,swap]{$\lambda$}	(h5);

%----------------------------------------------------------------------------------------

	\node 	(g')		at	(0,11.5) 	[circle]{$s_2$};

	\node 	(g4)		at	(0,10) 	[circle, thick, draw=red!50,fill=blue!20,inner sep=0pt,minimum size=0.40cm]{};

	\node 	(g5)		at	(0,8) 	[circle, thick, draw=red!50,fill=blue!20,inner sep=0pt,minimum size=0.40cm]{};

	\node 			at	(0,8) 	[circle, thick, draw=black!50,inner sep=0pt,minimum size=0.50cm]{};

	\draw[->, very thick]	(g')	to	node[auto,swap]{}	(g4);

	\draw[->, thick]	(g4)	to	node[auto,swap]{$!b$}	(g5);

	\node 	(h')		at	(2,11.5) 	[circle]{$s_3$};

	\node 	(h4)		at	(2,10) 	[circle, thick, draw=red!50,fill=blue!20,inner sep=0pt,minimum size=0.40cm]{};

	\node 	(h5)		at	(2,8) 	[circle, thick, draw=red!50,fill=blue!20,inner sep=0pt,minimum size=0.40cm]{};

	\node 			at	(2,8) 	[circle, thick, draw=black!50,inner sep=0pt,minimum size=0.50cm]{};

	\draw[->, very thick]	(h')	to	node[auto,swap]{}	(h4);

	\draw[->, thick]	(h4)	to	node[auto,swap]{$?b$}	(h5);
%----------------------------------------------------------------------------------------
	\draw[->, bend right=30, very thin]	(g5)	to	node[auto,swap]{$\lambda$}	(h5);

%----------------------------------------------------------------------------------------
	\node 	(h)		at	(3,12) 	[circle]{};
	\node 	(h')		at	(3,6) 	[circle]{};
	\draw[-	, thick]	(h)	to	node[auto]{}	(h');
%---------------------------------------------------------------------------------------------

	\node 	(g')		at	(4.5,11.5) 	[circle]{$s_0$};

	\node 	(g4)		at	(4.5,9) [rectangle, thick, draw=red!50,fill=blue!20,inner sep=0pt,minimum size=0.20cm]{};
	\node 			at	(4,9) [circle, minimum size=0.20cm]{$!a$};

	\node 	(h')		at	(6,11.5) 	[circle]{$s_1$};

	\node 	(h4)		at	(6,9) [rectangle, thick, draw=red!50,fill=blue!20,inner sep=0pt,minimum size=0.20cm]{};
	\node 			at	(6.5,9) [circle, minimum size=0.20cm]{$?a$};

	\draw[->, thick]	(g4)	to	node[auto]{}	(h4);

	\node 	(g')		at	(8.5,11.5) 	[circle]{$s_2$};

	\node 	(g4)		at	(8.5,9) [rectangle, thick, draw=red!50,fill=blue!20,inner sep=0pt,minimum size=0.20cm]{};
	\node 			at	(8,9) [circle, minimum size=0.20cm]{$!b$};

	\node 	(h')		at	(9.5,11.5) 	[circle]{$s_3$};

	\node 	(h4)		at	(9.5,9) [rectangle, thick, draw=red!50,fill=blue!20,inner sep=0pt,minimum size=0.20cm]{};
	\node 			at	(10,9) [circle, minimum size=0.20cm]{$?b$};

	\draw[->, thick]	(g4)	to	node[auto]{}	(h4);

\end{tikzpicture}
\end{center}
\caption{Sample Implementation of Choreography $C_1$ in \ref{fig:un-chor} and its Partial Order Execution}
\label{fig:un-impl}
\end{figure*}
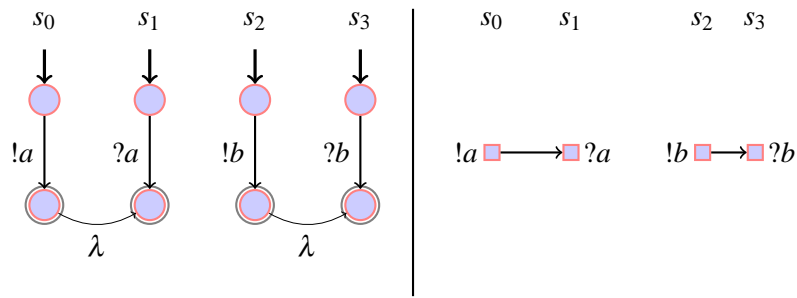

%Obviously $s_3$ has no way of knowing whether
%$a$ is sent. Thus $C_1$ is not realizable.% Such a ``missing connection''
%is a frequently cited reason for non-realizability \cite{Carbone06}.  %
%The set of service implementations $I_1$, generated by projection, for
%the choreography $C_1$ is shown in Figure \ref{fig:un-impl}. Note
%that $I_1$ does not realize $C_1$ because if $\calL$ is the set of all
%partial order executions of $I_1$, actually there is only one, then
%$Chor(\calL)\ne L(C_1)$. In fact, in any set of service implementations
%$I$ we define for $C_1$, there will never be any coupling between
%$!a$ and $!b$ transitions or $?a$ and $!b$ transitions. Consequently,
%there is no causal dependence between these pairs of events. Therefore,
%$Chor(\calL)$ will always contain the sequence $b\cdot a$ as well as
%the sequence $a \cdot b$.  \input{fig6}

We can define choreographies where the reasons for non-realizability are
not so obvious. Consider the choreography $C_2$, with $L(C_2) =\{a \cdot
b \cdot c, b \cdot a\}$, over three services $s_4,s_5,s_6$ and messages
$a,b,c$  shown in Figure \ref{fig:un-chor}. Since every service has
a FIFO queue, it can be shown that every implementation that permits
the two sequences in $L(C_2)$ will also permit the sequence ``$b \cdot
a \cdot c$''  that is not in $C_2$ \cite{Fu}.

Note that these unrealizable choreographies can be easily specified
by formulas of a standard temporal logic such as LTL, the temporal
logic of linear time. The following formula specifies the choreography $C_1$:

\[\Diamond (!a_0^1 \land \Diamond !b_2^3) \land 
\Box (!b_2^3 \Imply \Box \lnot !a_0^1).\]

Above, we have used the natural encoding $!a_i^j$ to denote a send
event from service $i$ to service $j$. It says that send-to-$s_1$ event in $s_0$
happens before send-to-$s_4$ event in $s_3$. It comes with an extra sanity check: 
there is no  send-to-$s_1$ event in $s_0$ after send-to-$s_4$ event in $s_3$.
On sequences, this is a satisfiable formula. But as we have already discussed,
such a choreography specification is unrealizable. This forms our
motivation for considering a local temporal logic on partial orders
where such specifications are unsatisfiable.

\section{Logic}
The logical language which we use to specify choreographies is a {\bf local} temporal logic. It is named as $p$-LTL.

\subsection{Syntax and Semantics} 
We fix the set of $n$ services $Ag=\{s_1,s_2,\cdots,s_n\}$. Further, we fix countable sets of {\it propositional letters} $P_s$, for local properties of service $s\in Ag$ and $\calM$ as the countable set of message symbols. The propositional symbols in $P_s$ are intended to specify internal actions in web service $s$. We assume, for convenience, that $P_s \cap P_{s'} = \emptyset$ for $s \neq s' \in Ag$. %Let $P \defn \cupover_{s \in Ag} P_s$.

The syntax of {\it $s$-local formulas}, local service formulas, is given below:
\[\alpha \in \Phi_s ::= !a_{s'}, a \in \calM \mid ?a_{s'}, a \in \calM \mid p \in P_s \mid \Not\alpha \mid \alpha_1 \lor \alpha_2 \mid \nxt \alpha \mid \Diamond\alpha \mid \ominus \alpha\]
$\nxt $ is the {\bf next}, $\Diamond$ is the {\bf eventual} and, $\ominus$ is the {\bf previous} temporal modality. $!a_{s'}$ is a send-$a$-to-$s'$ proposition in $s$ whereas $!a_{s'}$ is the corresponding receive-$a$-from-$s'$ proposition.

{\it Global formulas} are obtained by boolean combination of local formulas:
\[\psi \in \Psi ::= \alpha@s,~\alpha~\in~\Phi_s, s \in Ag\mid\Not~\psi\mid\psi_1~\lor~\psi_2\]
The propositional connectives $\land, \Imply, \equiv,\xor$ and derived 
temporal modality $\Box$ are defined as usual. In particular, 
$\Box\alpha\equiv \lnot\Diamond\lnot \alpha$. Fix $p_0 \in P_s$ and
let $\True = p_0 \lor \Not p_0$; let $\False = \Not \True$.

A choreography is a global formula $\psi \in \Psi$, that intuitively starts 
the services off in a global state, and their local dynamics and interactions are 
given by service formulas. Note that global safety properties can be specified by a
conjunction of local safety properties.

The formulas are interpreted on a class of partial orders, defined as follows. %For technical convenience, we consider only infinite behaviours. Formally, models are $2^P$-labelled Lamport diagrams defined as follows. 

$M=(E_{s_1},\cdots,E_{s_n}, \le,V)$ such that:
\begin{itemize}
\item $E_{s_1},\cdots,E_{s_n}$ are finite nonempty sets of events. 
$E_{s_i}$ is the set of events associated with service $s_i$. We assume that 
there is a unique event $\bot_{s_i} \in E_{s_i}$ for each $i$.  Let 
$E =\cupover_{i \in [n]}E_{s_i}$. 

\item $\le \subseteq (E \times E)$ is a partial order. Let $\lessdot$ be
the one-step relation derived from $\le$, that is: $\le = \lessdot^*$.

Define, for each service $s_i$, $\le_{s_i} = \le \cap (E_{s_i} \times 
E_{s_i})$. It gives the local behaviour of $s_i$. We require that 
$\bot_{s_i}$ be the unique minimum event in $E_{s_i}$; that is, for 
all $e \in E_{s_i}$, we have: $\bot_{s_i} \le_{s_i} e$.  Let 
$\lessdot_{s_i}$ be the one-step relation induced by $\le_i$.

Now define, across services, $\le_c \subseteq E \times  E$ by:
$\le_c = \{(e,e') \mid, e \lessdot e', e \in E_{s_i}, e' \in E_{s_j},
i \neq j, e \neq \bot_{s_i} \}$. It gives the global communication pattern 
among the services: interpret $(e,e')$ above as $s_i$ sending a message to $s_j$ 
with $e$ being the send event and $e'$ being the corresponding receive event. 
Note that the initial event cannot be a communication event.

\item $V:E \to 2^{(P\cup \calM)}$ such that for all $e \in E_s$, $V(e)\subseteq 
(P_s \cup \calM)$ gives the label of the event $e$, that is, the propositions 
that hold after the execution of the event $e$ and the messages that have been 
sent or received. 
\end{itemize}
%Given the local relations $\le_{s_i}$'s and the communication relation $<_c$ we define the global relation $\le$ as the reflexive transitive closure of $((\bigcup_i \le_{s_i}) \cup <_c)$.

Given a Lamport diagram $M=(E_{s_1},\cdots,E_{s_n}, \le,V)$, we can define 
the set of all configurations (global states) of $M$ as $\calC_M \subseteq 
E_{s_1} \times \cdots \times E_{s_n}$ such that every 
$c=(e_1,\cdots, e_n) \in \calC_M$ satisfies the following consistency property:
$\forall i,j, \forall e \in E_j ~\mbox{if}~ e \le e_i ~\mbox{then}~ e \le e_j$.
Thus each configuration is a tuple of local states of services. Note that
there is a unique initial global configuration $(\bot_1, \ldots, \bot_n) 
\in \calC_M$.

Let $\alpha \in \Phi_s$ and $e \in E_s$. The notion that $\alpha$ {\it is 
true at the event} $e$ of service $s$ in model $M$ is denoted 
$M, e \models_s \alpha$, and is defined inductively as follows:

\begin{itemize}

 \item $M, e \models_s p$ iff $p \in V(e)$.

 \item $M, e \models_s !a_{s'}$ iff $\exists e' \in E_{s'}$ such that $(e,e') \in <_c$ and $a \in V(e)$.

 \item $M, e \models_s ?a_{s'}$ iff $\exists e' \in E_{s'}$ such that $(e',e) \in <_c$ and $a \in V(e)$.

 \item $M, e \models_s \Not \alpha$ iff $M, e \not\models_s \alpha$. 

 \item $M, e \models_s \alpha \lor \alpha'$ iff  $M, e \models_s \alpha$ or $M, e \models_s \alpha'$.

 \item $M, e \models_s \nxt \alpha$ iff there exists $e' \in E_s$ such that $e \lessdot_s e'$ and $M, e' \models_s \alpha$.

 \item $M, e \models_s \Diamond \alpha$ iff $\exists e' \in E_s$: $e \le_s e', M, e' \models_s \alpha$.

 \item $M, e \models_s \ominus \alpha$ iff there exists $e' \in E_s$ such that $e' \lessdot_s e$ and $M, e' \models_s \alpha$. 

\end{itemize}

When the send proposition $!a_{s'}$ holds at $e$ in service $s$, it means 
there is a corresponding receive event $e'$ in service $s'$ ($(e,e') \in <_c$) and $a$ holds 
locally in $e$. Similarly, when the receive proposition $?a_{s'}$ holds 
at $e$ in service $s$, it means there is a corresponding send event $e'$ 
in service $s'$ ($(e'e)\in <_c$) and $a$ holds locally in $e$.

Also, when $\nxt \alpha$ holds at $e$ in service $s$, it means that $\alpha$
holds at $e'$, the one-step successor of $e$. Similarly, when $\ominus \alpha$ 
holds at $e$ in service $s$, it means that $\alpha$ holds at $e'$, the one-step predecessor of $e$.
Clearly, we see that $M, e \models_s \ominus \False$ iff $e = \bot_s$.
Further, when $\Diamond \alpha$ holds at $e$ in service $s$, it means that $\alpha$
holds at $e'$, a descendant of $e$.

For every global state $c=(e_1,\cdots, e_n) \in \calC_M$ and global formula 
$\psi \in \Psi$, we define global satisfiability $M,c \models \psi$  ($\psi$ 
{\it is true at configuration} $c$ of the model $M$) inductively as follows:

\begin{itemize}

 \item $M, c \models \alpha@s_i$ iff $M, e_{i} \models_{s_i} \alpha$.

 \item $M, c \models \Not\psi$ iff $M, c \not\models \psi$.

 \item $M, c \models \psi_1 \lor \psi_2$ iff $M, c \models \psi_1$ or $M, c \models \psi_2$.

\end{itemize}

Given a choreography $\psi$ in $p$-LTL we define the set $Models(\psi)$ as all 
the Lamport diagrams $M$ such that $M,c_0 \models \psi$, where $c_0$ is the
unique initial global configuration of $M$.
%We say that $\Box \psi$ is satisfiable if $Models(\Box \psi) \not = \emptyset$.

%Given a Lamport diagram $M=(E_{s_1},\cdots,E_{s_n}\le_{s_1},\cdots \le_{s_n},\le_c,V)$ over a set of agents $Ag$, we add an empty event $\bot_s$ to event set of each service $s\in Ag$, for technical convenience. Let $E_s'= E_s \cup \{\bot_s\}$ and $E'=\cupover_{s \in Ag} E_s'$. We extend $\lessdot_{s}$ by adding an extra tuple $(\bot_{s},e_s^0)$, where $e_s^0$ is the minimum event in $E_s$ with respect to $\lessdot_s$. There is no change in $<_c$. The partial order over $E'$, $\le$, is defined as the reflexive transitive closure of $((\bigcup_i \le_{s_i}) \cup <_c)$, as given above.  The set of configurations defined over $E_{s_i}'$'s is denoted by $\calC_M'$.

\subsection{Choreography Examples}

The simplest choreography which can be encoded using $p$-LTL is that of producer-consumer protocol. This protocol describes the behaviour of two services, producer ($p$) and consumer ($c$), in which $p$ produces objects labelled $a$ which are consumed by $c$. The objects produced by $p$ are put into a FIFO buffer from where $c$ retrieves them. Clearly, the protocol can be modelled as an asynchronous message passing system, where $p$ sends messages labelled $a$ to $c$. Depending on the size of buffer, there are various patterns of messages exchanged between $p$ and $c$. Figure \ref{fig:ld-pc} gives the scenarios for the cases where buffer size is $1$, $2$ and $3$.

The choreography for producer-consumer protocol may be formulated as $\psi$ where $\psi \defn \Box !a_c @ p \land \Box ?a_p @c$. It can be seen that Lamport diagrams in \ref{fig:ld-pc} are legitimate models of $\psi$.

Let us consider another choreography example. This concerns a system comprising three services: a traveller ($T$) and map providers ($M_1$ and $M_2$). The GPS device of traveller ($T$) has to automatically negotiate a purchase agreement with one of the two map providers. After $T$ has already broadcast a ``request of bid'' message, the two services $M_1$ and $M_2$ send their respective bids. $T$ evaluates the two bids and accepts one.

The message set is fixed as $\calM=\{bid,bid',acc,rej\}$. We assume $rej\equiv \lnot acc$. The choreography can be formulated as $\psi\defn \alpha @ M_1\land \beta @ M_2 \land \gamma @ T$ and:
 
\begin{itemize}
\item $\alpha \defn \Box \big(!bid_T \Imply \Diamond (?acc_T \lor ?rej_T)\big)$

\item $\beta \defn  \Box \big(!bid'_T \Imply \Diamond (?acc_T \lor ?rej_T)\big)$

\item $\gamma \defn \Box \big((?bid_{M_1} \Imply \Diamond (!acc_{M_1} \lor !rej_{M_1})) \land (?bid_{M_2} \Imply \Diamond (!acc_{M_2} \lor !rej_{M_2})) \land (\Diamond !acc_{M_1} \Imply  \Diamond !rej_{M_2})\land (\Diamond !acc_{M_2} \Imply  \Diamond !rej_{M_1})\big)$
\end{itemize}

We briefly explain the local formulas: $\alpha$ says that when a bid is send to $T$ (by $M_1$), it eventually receives either an acceptance or rejection. $\beta$ says the same for $M_2$. $\gamma$ says two things: when $T$ receives a bid from $M_1$ ($M_2$) it either accepts or rejects it and, exactly one of the bids ($bid$ or $bid'$) is accepted.

The logical formalism which we have introduced in this section can not specify unrealizable choreographies of the kind mentioned in the previous section. Further, it is expected that software designers will not learn to write such formulas but will use tools that work with graphical formalisms and generate specifications interactively.

\section{System of Communicating Automata}
The service implementations for choreographies are given in terms of Systems of Communicating Automata (SCA). SCAs are quite similar to the automata model introduced in \cite{MR}.

We fix $n > 0$ and focus our attention on $n$-service systems. Let $[n]=\{1,2,\cdots,n\}$. A {\bf distributed alphabet} for such systems is an $n$-tuple $\DA$, where for each $i \in [n]$, $\Sigma_i$ is a finite non-empty alphabet of actions of service $i$ and for all $i \neq j$, $\Sigma_i \cap \Sigma_j = \emptyset$. The alphabet induced by $\DA$ is given by $\Sigma = \cupover_i \Sigma_i$. The set of {\bf system actions} is the set $\Sigma' = \{\lambda\} \cup \Sigma$. The action symbol $\lambda$ is referred to as the {\bf communication action}. This is used as an action representing a communication constraint through which every receive action will be dependent on its corresponding send action. We use $a, b, c$ {\it etc.}, to refer to elements of $\Sigma$ and $\tau, \tau'$ {\it etc.}, to refer to those of $\Sigma'$.

\begin{dfn} 
A {\bf System of $n$ Communicating Automata (SCA)} on a distributed alphabet 
\newline
$\DA$ is a tuple $S = ((Q_1,,F_1), \ldots, (Q_n,F_n), \to, Init)$ where,

\begin{enumerate}

  \item For each $j \in [n]$, $Q_j$ is a finite set of (local) states of service $j$. 
  \newline 	
  For $j \neq j'$, $Q_j \cap Q_{j'} = \emptyset$. 

  \item for each $j \in [n]$, $F_j \subseteq Q_j$ is the set of
  (local) final states of service $j$.

  \item Let $Q = \cupover_j Q_j$, then, the transition relation $\to$
  is defined over $Q$ as follows. $\to \subseteq (Q \times \Sigma' 
  \times Q)$ such that if $q \step{\tau} q'$ then either there exists 
  $j$ such that $\{q, q'\} \subseteq Q_j$ and $\tau \in \Sigma_j$, or 
  there exist $j \neq j'$ such that $q \in Q_j, q' \in Q_{j'}$ and 
  $\tau = \lambda$. 

  %\item for each $j \in [n]$, $I_j \subseteq Q_j$ is the set of 
  %(local) initial states of the system. 

  \item $Init \subseteq (Q_1 \times \cdots \times Q_n)$ is 
the set of global initial states of the system.
  
%  \item for each $j \in [n]$, $\to_i \subseteq Q_j \times \Sigma_j \times Q_j$ is the local transition relation for service $j$.

\end{enumerate}
\end{dfn}

Thus, SCAs are systems of $n$ finite state automata with $\lambda$-labelled
communication constraints between them. The only `global' specification
is on initial states. This is in keeping with design of choreographies:
the services are `set up' and once initiated, manage themselves without
global control.

Note that $\to$ above is {\em not} 
a global transition relation, it consists of {\bf local transition 
relations}, one for each service, and {\bf communication constraints} 
of the form $q \step{ \lambda} q'$, where $q$ and $q'$ are states of 
different services. The latter define a coupling relation rather than 
a transition. The interpretation of local transition relations is standard: 
when the service $i$ is in state $q_1$ and reads input $a \in \Sigma_i$, it 
can move to a state $q_2$ and be ready for the next input if $(q_1, a, q_2) 
\in \to$. 

The interpretation of communication constraints is non-standard 
and depends only on automaton states, not on local input. When $q \step{ \lambda}
q'$, where $q \in Q_i$ and $q' \in Q_j$, it constrains the system 
behaviour as follows: whenever service $i$ is in state $q$, it puts
a message whose content is $q$ and intended recipient is $j$ into the 
buffer; whenever service $j$ intends to enter state $q'$, it checks its 
environment to see if a message of the form $q$ from $i$ is available for 
it, and waits indefinitely otherwise. If a system $S$ has no $\lambda$ 
constraints at all, automata  proceed asynchronously and do not 
wait for each other. We will refer to $\lambda$-constraints as 
`$\lambda$-transitions' in the sequel for uniformity, but this 
explanation (that they are constraints not dependent on local input) should
be kept in mind. 

\begin{figure*}
\begin{center}
\begin{tikzpicture}[scale=.5]
	\node 	(e0)		at	(-5,5) 	[rectangle, thick]{$p$};

	\node 	(e1)		at	(-5,4) 	[rectangle, thick, draw=blue!50,fill=green!20,inner sep=0pt,minimum size=.35cm]{$e_1$};
	\node 			at	(-6,4) 	[rectangle, thick, inner sep=0pt,minimum size=.35cm]{$!a$};

%---------------------------------------------------------------------------------------------------------------------------
	\node 	(f0)		at	(-1,5) 	[rectangle, thick]{$c$};

	\node 	(f1)		at	(-1,4) 	[rectangle, thick, draw=blue!50,fill=green!20,inner sep=0pt,minimum size=.35cm]{$f_1$};
	\node 			at	(0,4) 	[rectangle, thick, inner sep=0pt,minimum size=.35cm]{$?a$};

	\node 			at	(-3,-2) [rectangle, thick]{$(i)$};

%----------------------------------------------------------------------------------------------------------------------------------
	\draw[->, very thick]	(e1)	to	node[auto]{}	(f1);

%-----------------------------------------------------------------------------------------------------------------------------------
%two send events
	\node 	(h)		at	(1,6) 	[rectangle, thick]{};
	\node 	(h')		at	(1,-1) 	[rectangle, thick]{};
	\draw[-, thick]	(h)	to	node[auto]{}	(h');

%----------------------------------------------------------------------------------------------------------------------------------
	\node 	(e0)		at	(3,5) 	[rectangle, thick]{$p$};

	\node 	(e1)		at	(3,4) 	[rectangle, thick, draw=blue!50,fill=green!20,inner sep=0pt,minimum size=.35cm]{$e_1$};
	\node 			at	(2,4) 	[rectangle, thick, inner sep=0pt,minimum size=.35cm]{$!a$};

	\node 	(e2)		at	(3,2) 	[rectangle, thick, draw=blue!50,fill=green!20,inner sep=0pt,minimum size=.35cm]{$e_2$};
	\node 			at	(2,2) 	[rectangle, thick, inner sep=0pt,minimum size=.35cm]{$!a$};

	\draw[->, thick]	(e1)	to	node[auto]{}	(e2);

%---------------------------------------------------------------------------------------------------------------------------
	\node 	(f0)		at	(7,5) 	[rectangle, thick]{$c$};

	\node 	(f1)		at	(7,4) 	[rectangle, thick, draw=blue!50,fill=green!20,inner sep=0pt,minimum size=.35cm]{$f_1$};
	\node 			at	(8,4) 	[rectangle, thick, inner sep=0pt,minimum size=.35cm]{$?a$};

	\node 	(f2)		at	(7,2) 	[rectangle, thick, draw=blue!50,fill=green!20,inner sep=0pt,minimum size=.35cm]{$f_2$};
	\node 			at	(8,2) 	[rectangle, thick, inner sep=0pt,minimum size=.35cm]{$?a$};

	\draw[->, thick]	(f1)	to	node[auto]{}	(f2);

	\node 			at	(5,-2) [rectangle, thick]{$(ii)$};

%----------------------------------------------------------------------------------------------------------------------------------
	\draw[->, very thick]	(e1)	to	node[auto]{}	(f1);
	\draw[->, very thick]	(e2)	to	node[auto]{}	(f2);

%-----------------------------------------------------------------------------------------------------------------------------------
%three send events
	\node 	(h)		at	(9,6) 	[rectangle, thick]{};
	\node 	(h')		at	(9,-1) 	[rectangle, thick]{};
	\draw[-, thick]	(h)	to	node[auto]{}	(h');

%----------------------------------------------------------------------------------------------------------------------------------
	\node 	(e0)		at	(11,5) 	[rectangle, thick]{$p$};

	\node 	(e1)		at	(11,4) 	[rectangle, thick, draw=blue!50,fill=green!20,inner sep=0pt,minimum size=.35cm]{$e_1$};
	\node 			at	(10,4) 	[rectangle, thick, inner sep=0pt,minimum size=.35cm]{$!a$};

	\node 	(e2)		at	(11,2) 	[rectangle, thick, draw=blue!50,fill=green!20,inner sep=0pt,minimum size=.35cm]{$e_2$};
	\node 			at	(10,2) 	[rectangle, thick, inner sep=0pt,minimum size=.35cm]{$!a$};

	\node 	(e3)		at	(11,0) 	[rectangle, thick, draw=blue!50,fill=green!20,inner sep=0pt,minimum size=.35cm]{$e_3$};
	\node 			at	(10,0) 	[rectangle, thick, inner sep=0pt,minimum size=.35cm]{$!a$};

	\draw[->, thick]	(e1)	to	node[auto]{}	(e2);
	\draw[->, thick]	(e2)	to	node[auto]{}	(e3);

%---------------------------------------------------------------------------------------------------------------------------
	\node 	(f0)		at	(15,5) 	[rectangle, thick]{$c$};

	\node 	(f1)		at	(15,4) 	[rectangle, thick, draw=blue!50,fill=green!20,inner sep=0pt,minimum size=.35cm]{$f_1$};
	\node 			at	(16,4) 	[rectangle, thick, inner sep=0pt,minimum size=.35cm]{$?a$};

	\node 	(f2)		at	(15,2) 	[rectangle, thick, draw=blue!50,fill=green!20,inner sep=0pt,minimum size=.35cm]{$f_2$};
	\node 			at	(16,2) 	[rectangle, thick, inner sep=0pt,minimum size=.35cm]{$?a$};

	\node 	(f3)		at	(15,0) 	[rectangle, thick, draw=blue!50,fill=green!20,inner sep=0pt,minimum size=.35cm]{$f_3$};
	\node 			at	(16,0) 	[rectangle, thick, inner sep=0pt,minimum size=.35cm]{$?a$};

	\draw[->, thick]	(f1)	to	node[auto]{}	(f2);
	\draw[->, thick]	(f2)	to	node[auto]{}	(f3);

	\node 			at	(13,-2) [rectangle, thick]{$(iii)$};

%----------------------------------------------------------------------------------------------------------------------------------
	\draw[->, very thick]	(e1)	to	node[auto]{}	(f1);
	\draw[->, very thick]	(e2)	to	node[auto]{}	(f2);
	\draw[->, very thick]	(e3)	to	node[auto]{}	(f3);

\end{tikzpicture}
\end{center}
\caption{Lamport diagrams of the producer-consumer protocol}
\label{fig:ld-pc}
\end{figure*}
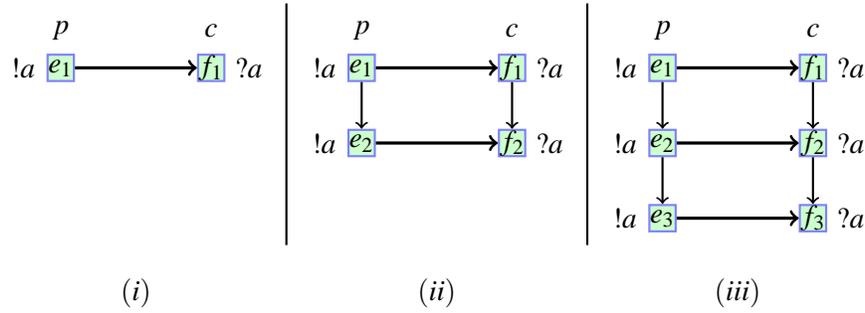

We  use the notation $\presca{q} \defn \{q' \mid q' \step{\lambda}
q\}$ and $\postsca{q} \defn \{q' \mid q \step{\lambda} q'\}$. For $q \in Q$, 
the set $\presca{q}$ refers to the set of all states from which $q$ has 
incoming $\lambda$-transitions and the set $\postsca{q}$ is the set of all 
states to which $q$ has outgoing $\lambda$-transitions. 
The {\em global behaviour} of an SCA will be defined using its set of 
{\bf global states} $\widetilde{Q}  = Q_1 \times \cdots \times Q_n $. 
When $\widetilde{q} = (q_1, \dots, 
q_n) \in \widetilde{Q}$, we use the notation $\widetilde{q}[i]$ to refer to $q_i$.
The language accepted by an SCA is a collection of ($\Sigma$-labelled) 
Lamport diagrams, to be defined below.

Figure~\ref{fig:sca-eg} gives an SCA over the alphabet $\Sigtil = 
(\{!a\}, \{?a\})$. 
The (global) initial state of this SCA is $\{(q_0,q_0')\}$ and the (global) final state is $\{(q_2,q_2')\}$
The reader will observe that this SCA models the producer-consumer protocols given in Figure \ref{fig:ld-pc}.

The producer generates the first object via the $q_0\step{!a}q_1$ transition, any number of objects via the  $q_1\step{!a}q_1$ transition, and the last object via the $q_1\step{!a}q_2$ transition. The consumer consumes the first object via the $q_0'\step{?a}q_1'$ transition, any number of objects via the  $q_1'\step{?a}q_1'$ transition, and the last object via the $q_1'\step{?a}q_2'$ transition. As a consumption can follow only after a production there is a $\lambda$ transition between $q_0$ and $q_1'$ and  $q_0$ and $q_2'$ and also between $q_1$ and $q_2'$. %Notice that there is no incoming transition in the local initial states $q_0$ and $q_0'$ and no outgoing transitions from local final states $q_2$ and $q_2'$.

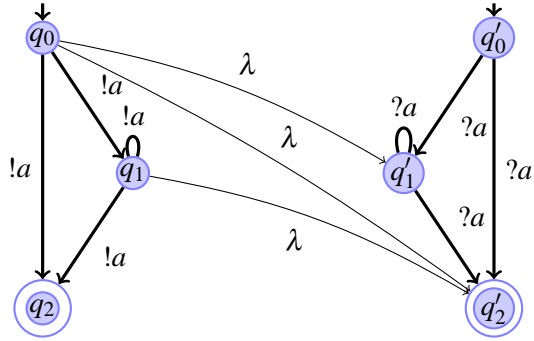
\begin{figure*}
\begin{center}
\begin{tikzpicture}[scale=0.30]
	\node 	(e0)		at	(0,2) 	[rectangle, thick]{};
	\node 	(s0)		at	(0,0) 	[circle, thick, draw=blue!50,fill=blue!20,inner sep=0pt,minimum size=.3cm]{$q_0$};

	\node 	(s1)		at	(4,-6) 	[circle, thick, draw=blue!50,fill=blue!20,inner sep=0pt,minimum size=.3cm]{$q_1$};

	\node 			at	(0,-12) 	[circle, thick, draw=blue!50,fill=blue!20,inner sep=0pt,minimum size=.3cm]{$q_2$};
	\node 	(s2)		at	(0,-12) 	[circle, thick, draw=blue!50,inner sep=0pt,minimum size=.75cm]{};

	\draw[->,very thick] 		(e0) to 	node[auto]{} 			(s0);
	\draw[->,very thick] 		(s0) to 	node[auto]{$!a$} 		(s1);
	\draw[->,very thick] 		(s0) to 	node[auto,swap]{$!a$} 		(s2);
	\draw[->,very thick] 		(s1) to 	node[auto]{$!a$} 		(s2);
	\draw[loop above,very thick] 	(s1) to 	node[auto]{$!a$} 		(s1);

%------------------------------------------------------------------------------------------------------------------------------------
	\node 	(f0)		at	(20,2) 	[rectangle, thick]{};

	\node 	(t0)		at	(20,0) 	[circle, thick, draw=blue!50,fill=blue!20,inner sep=0pt,minimum size=.3cm]{$q_0'$};

	\node 	(t1)		at	(16,-6) 	[circle, thick, draw=blue!50,fill=blue!20,inner sep=0pt,minimum size=.3cm]{$q_1'$};

	\node 			at	(20,-12) 	[circle, thick, draw=blue!50,fill=blue!20,inner sep=0pt,minimum size=.3cm]{$q_2'$};
	\node 	(t2)		at	(20,-12) 	[circle, thick, draw=blue!50,inner sep=0pt,minimum size=.75cm]{};

	\draw[->,very thick] 		(f0) to 	node[auto]{} 			(t0);
	\draw[->,very thick] 		(t0) to 	node[auto]{$?a$} 		(t1);
	\draw[->,very thick] 		(t0) to 	node[auto]{$?a$} 		(t2);
	\draw[->,very thick] 		(t1) to 	node[auto]{$?a$} 		(t2);
	\draw[loop above,very thick] 	(t1) to 	node[auto]{$?a$} 		(t1);
%-----------------------------------------------------------------------------------------------------------
	\draw[->,thin, bend left=10]	(s1)	to	node[auto,swap]{$\lambda$}	(t2);
	\draw[->,thin, bend left=5]	(s0)	to	node[auto]{$\lambda$}	(t2);
	\draw[->,thin, bend left=10]	(s0)	to	node[auto]{$\lambda$}	(t1);
%------------------------------------------------------------------------------------------------------------
%	\node 			at	(10,-16) 	[rectangle, thick]{$\widetilde{Q}=\{(q_1,q_1'),(q_1,q_2')(q_2,q_1')(q_2,q_2')\}$};

\end{tikzpicture}
\end{center}
\caption{A simple SCA}
\label{fig:sca-eg}
\end{figure*}

\subsection{Poset language of an SCA}
We now formally define the run of an SCA on its input, a Lamport diagram 
and the poset language accepted by an SCA as the collection of Lamport 
diagrams on which the SCA has an accepting run. 

Given an SCA $S$ on $\Sigtil$, a {\bf run} of $S$ on a Lamport
diagram $D=(E_1,\cdots,E_n, \le,V)$ is a map $\rho : \calC_D \rightarrow 
\widetilde{Q}$, $V:E \to \Sigma$, 
such that the following conditions are satisfied:

\begin{itemize}

  \item $\rho((\bot_1,\cdots,\bot_n)) \in Init$. 

  \item For $c \in \calC_D$, suppose $\rho(c) = (q_1, q_2, \ldots, q_n)$. 
  Consider $c'\in \calC_D$, such that $c$ differs from $c'$ only at the $i$th 
position. Let $e \in E_i$ be the $i$th element in $c'$ and 
  $V(e) = \sigma\in \Sigma_i$. Then,

  \begin{itemize}

   \item $\rho(c') = (q'_1, q'_2, \ldots, q'_n)$ where $q'_j = q_j$ for all 
   $j \neq i$ and $q_i \step{\sigma} q'_i$ in $S$. 

   \item Suppose $\exists e'\in E_j$, $j\ne i$ such that $e'<_c e$. Let $f' \in E_j'$ such that $f' \lessdot_j e'$.  Let $c_0,c_1 \in \calC_D'$  such that the $j$th element in $c_0$ be $f'$ and in $c_1$ be $e'$ whereas all the other elements are the same. Then,
	 $\rho(c_0)[j]\longstep{V(e')}\rho(c_1)[j]$ and $\rho(c_0)[j]\longstep{\lambda}\rho(c')[i]$. Remove $\rho(c_0)[j]$ from the front of 		 FIFO queue of $j$ for $i$, if it is there else block.

   \item If $\postsca{q_i} \cap Q_j \neq \emptyset$, then, there exists 
   $e' \in E_j$ such that $e \lessdot e'$. Insert $q_i$ in the FIFO queue for $j$.

  \end{itemize}

\end{itemize}

Thus, a run of $S$ on $D$ is a map from the set $\calC_D$ of configurations 
of $D$ to the set of global states of $S$ such that the following conditions
hold: If $c'$ is a configuration obtained by adding an event $e \in E_i$
(where $V(e)=\sigma$) to a configuration $c$ then, there is a transition 
on $\sigma$ from the local state of service $i$ in $\rho(c)$ to the local state of 
the same service in $\rho(c')$ and all other local states are unaltered. In 
addition, if $e$ is a receive event, we ensure that the corresponding 
send event has already occurred and that there is a $\lambda$-constraint into
the resulting state. When there are out-going $\lambda$-constraints
from the target state of the enabling transition, note that the definition makes sure that the corresponding event
 $e$ is a send event and that it has a matching receive event. 

Now, we specify the acceptance condition for a run $\rho:\calC_D\rightarrow \widetilde{Q}$ of $S$ over $D$ . Let $c=(e_1,\cdots,e_n) \in \calC_D$ such that for each $i \in [n]$, $e_i$ is the $i$-maximal event in $D$. The run $\rho$ is said to be {\bf accepting} if for each $i \in [n]$, $\rho(c)[i] \in F_i$.
The {\bf poset language accepted by $S$} is denoted by $\posetlang{S}$ and 
is defined as:
$\posetlang{S} \defn \{ D \mid D ~\mbox{is a Lamport diagram and $S$ has an accepting run on $D$} \}.$ 

%----------------------------------------------------------------------------------------------------------------------------------------

%-----------------------------------------------------------------------------------------------------------------------------------------
For example, Figure ~\ref{fig:run-eg} gives a run of the SCA in Figure ~\ref{fig:sca-eg}
over the Lamport diagram $(iii)$ of producer-consumer problem given in Figure ~\ref{fig:ld-pc}.
The figure essentially gives the directed acyclic graph corresponding to the 
configuration space of the Lamport diagram. Each node (configuration) has an associated 
state label given in shaded boxes on the right. 

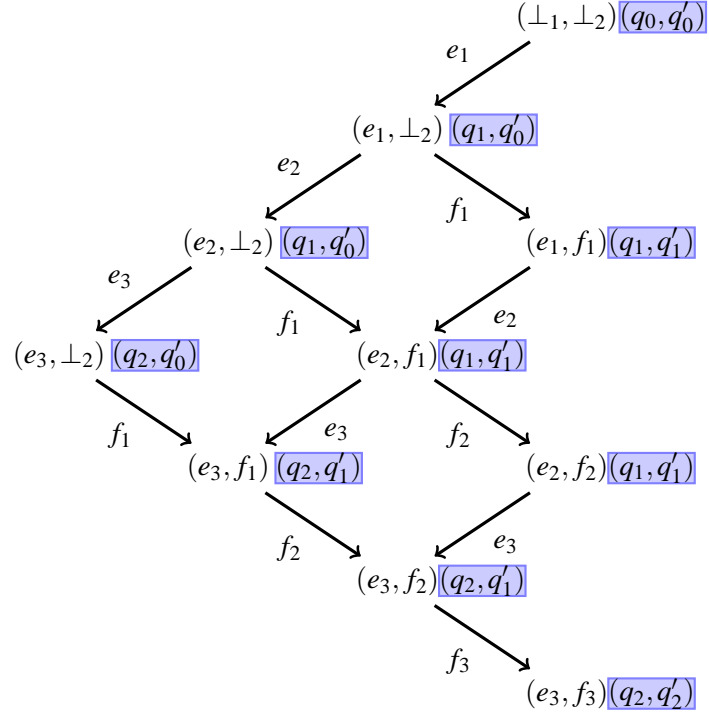
\begin{figure*}
\begin{center}
\begin{tikzpicture}[scale=0.75]
	\node 	(n0)		at	(3,2) 	[rectangle]{$(\bot_1,\bot_2)$};
	\node 			at	(4.7,2) [rectangle, thick, draw=blue!50,fill=blue!20,inner sep=0pt,minimum size=.4cm]{$(q_0,q_0')$};

	\node 	(n1)		at	(0,0) 	[rectangle]{$(e_1,\bot_2)$};
	\node 			at	(1.7,0) [rectangle, thick, draw=blue!50,fill=blue!20,inner sep=0pt,minimum size=.4cm]{$(q_1,q_0')$};

	\node 	(n21)		at	(-3,-2) [rectangle]{$(e_2,\bot_2)$};
	\node 			at	(-1.3,-2) [rectangle, thick, draw=blue!50,fill=blue!20,inner sep=0pt,minimum size=.4cm]{$(q_1,q_0')$};

	\node 	(n22)		at	(3,-2) 	[rectangle]{$(e_1,f_1)$};
	\node 			at	(4.5,-2) [rectangle, thick, draw=blue!50,fill=blue!20,inner sep=0pt,minimum size=.4cm]{$(q_1,q_1')$};

	\node 	(n31)		at	(-6,-4) [rectangle]{$(e_3,\bot_2)$};
	\node 			at	(-4.3,-4) [rectangle, thick, draw=blue!50,fill=blue!20,inner sep=0pt,minimum size=.4cm]{$(q_2,q_0')$};

	\node 	(n32)		at	(0,-4)  [rectangle]{$(e_2,f_1)$};
	\node 			at	(1.5,-4)  [rectangle, thick, draw=blue!50,fill=blue!20,inner sep=0pt,minimum size=.4cm]{$(q_1,q_1')$};

	\node 	(n41)		at	(-3,-6) [rectangle]{$(e_3,f_1)$};
	\node 			at	(-1.4,-6) [rectangle, thick, draw=blue!50,fill=blue!20,inner sep=0pt,minimum size=.4cm]{$(q_2,q_1')$};

	\node 	(n42)		at	(3,-6) 	[rectangle]{$(e_2,f_2)$};
	\node 			at	(4.5,-6) [rectangle, thick, draw=blue!50,fill=blue!20,inner sep=0pt,minimum size=.4cm]{$(q_1,q_1')$};

	\node 	(n5)		at	(0,-8) [rectangle]{$(e_3,f_2)$};
	\node 			at	(1.5,-8) [rectangle, thick, draw=blue!50,fill=blue!20,inner sep=0pt,minimum size=.4cm]{$(q_2,q_1')$};

	\node 	(n52)		at	(3,-10) [rectangle]{$(e_3,f_3)$};
	\node 			at	(4.5,-10) [rectangle, thick, draw=blue!50,fill=blue!20,inner sep=0pt,minimum size=.4cm]{$(q_2,q_2')$};

%	\node 	(f01)		at	(-9,-6) [rectangle]{};
%	\node 	(f02)		at	(-6,-8) [rectangle]{};

	\draw[->,very thick] 		(n0) to node[auto,swap]{$e_1$} 		(n1);

	\draw[->,very thick] 		(n1) to node[auto,swap]{$e_2$} 		(n21);
	\draw[->,very thick] 		(n1) to node[auto,swap]{$f_1$} 		(n22);

	\draw[->,very thick] 		(n21) to node[auto,swap]{$e_3$} 	(n31);
	\draw[->,very thick] 		(n21) to node[auto,swap]{$f_1$} 		(n32);
	\draw[->,very thick] 		(n22) to node[auto]{$e_2$} 		(n32);

	\draw[->,very thick] 		(n31) to node[auto,swap]{$f_1$} 	(n41);
	\draw[->,very thick] 		(n32) to node[auto]{$e_3$} 		(n41);
	\draw[->,very thick] 		(n32) to node[auto,swap]{$f_2$} 		(n42);

	\draw[->,very thick] 		(n42) to node[auto]{$e_3$} 		(n5);
	\draw[->,very thick] 		(n41) to node[auto,swap]{$f_2$} 		(n5);

	\draw[->,very thick] 		(n5) to node[auto,swap]{$f_3$} 		(n52);

%	\draw[->, thin] 		(n31) to node[auto,swap]{} 		(f01);
%	\draw[->, thin] 		(n41) to node[auto,swap]{} 		(f02);

\end{tikzpicture}
\end{center}
\caption{The run of SCA  in \ref{fig:sca-eg} over Lamport diagram $(iii)$ in \ref{fig:ld-pc}}
\label{fig:run-eg}
\end{figure*}

\section{Realizability Algorithm}
We now formulate the realizability problem for web service choreography in
our setting and show that it is decidable. We do this by the so-called 
automata-theoretic approach of model checking. A composite web service 
implementation $I$ is modelled as an SCA $S$ and a choreography $C$ is 
given by a formula $\psi$ in $p$-LTL. Given a choreography $\psi$, the 
realizability problem is to check if there exists a composite web service 
implementation $S$ that conforms to the choreography $\psi$ i.e, to check if 
the global ``behaviour'' of $S$ ``satisfies'' $\psi$. In order to do this, we give 
an algorithm to construct the system $S_{\psi}$ accepting the models of $\psi$.

\subsection{Formula Automaton for $p$-LTL}
In this section we show that one can effectively associate an SCA $S_{\psi}$ 
with each $p$-LTL choreography $\psi$ in such a way that the global behaviour 
of $S_{\psi}$ satisfies the formula $\psi$: that is, 
$\posetlang{S_{\psi}}= \it{Models}(\psi)$.

Given $\psi$, we first define a subformula closure set $CL_s$ for each 
$s \in Ag$. This set $CL_s$ of agent $s$ is used to define the local states 
of that agent. Given a global formula $\psi$, the set $CL(\psi)$ and $CL_s$ 
for $s \in Ag$, are defined by simultaneous induction to be the least set
of formulas such that:
\begin{enumerate}

  \item $\psi \in CL(\psi)$.

  \item if $\alpha@s \in CL(\psi)$ then  $\alpha \in CL_s$.

  \item $\{!a_{s'}, ?a_{s'}\} \subseteq CL_s$, for each 
$s' \in Ag$, $s \neq s'$.

  \item $True \in CL_s$; we take $\lnot True$ as $False$.  
$\nxt \False \in CL_s$ and $\ominus \False \in CL_s$.
  
  \item $\psi' \in CL(\psi)$ iff $\Not\psi' \in CL(\psi)$, taking 
$\Not\Not \psi$ to be $\psi$.  $\alpha \in CL_s$ iff $\Not\alpha \in CL_s$, 
taking $\Not\Not \alpha$ to be $\alpha$.

  \item if $\psi_1 \lor \psi_2 \in CL(\psi)$ then $\psi_1, \psi_2 \in CL(\psi)$.
  if $\alpha_1 \lor \alpha_2 \in CL_s$ then $\alpha_1, \alpha_2 \in CL_s$.

  \item if $\nxt \alpha \in CL_s$ then $\alpha \in CL_s$.

  \item if $\Diamond \alpha \in CL_s$ then $\alpha$, $\nxt(\Diamond \alpha) \in CL_s$.

  \item if $\ominus \alpha \in CL_s$ then $\alpha \in CL_s$.

\end{enumerate}

It can be checked that $|CL(\psi)|$ and each $|CL_s|$ are linear in the 
size of $\psi$. For the rest of this section, fix a global formula 
$\psi_0 \in \Psi$. 
We will refer to $CL(\psi_0)$ simply as $CL$ and $CL_s$ will refer 
to the associated sets of $s$-local formulas.  
We also use $U_s \defn \{\Diamond\alpha \mid \Diamond\alpha \in CL_s\}$.

We say that $A \subseteq CL_s$ is an {\it $s$-atom} iff it is locally
consistent, that is, it contains $\True$ and:

\begin{enumerate}

  \item for every formula $\alpha \in CL_s$, either $\alpha \in A$ or $\Not \alpha \in A$ but not both.

  \item for every formula $\alpha \lor \alpha' \in CL_s$, $\alpha \lor \alpha' \in A$ iff $\alpha \in A$ or $\alpha' \in A$.

  \item for every formula $\Diamond \alpha \in CL_s$, $\Diamond\alpha\in A$ iff $\alpha\in A$ or $\nxt(\Diamond\alpha) \in A$. 

%  \item for every formula $!a_{s_j} \in CL_s$, if $!a_{s_j} \in A$ then $a \in A$.

% \item for every formula $?a_{s_j} \in CL_s$, if $!a_{s_j} \in A$ then $a \in A$.

  \item if $\ominus \False \in A$ then for every $\ominus \beta \in CL_s$, $\ominus \beta \not \in A$.

\end{enumerate}
An $s$-atom $A$ is said to be {\em initial} if $\ominus False \in A$% and $\{\lnot !a_{s'},\lnot ?a_{s'}\}\subseteq A$ for all $s' \in Ag$, $s' \neq s$.

Let $AT_s$ denote the set of all $s$-atoms. Let $AT \defn \cupover_s AT_s$. Let $\widetilde{AT}$ denote the set $AT_1 \times \ldots
\times AT_n$. We let $\widetilde{A}, \widetilde{B}$ {\it etc.}, to range over $\widetilde{AT}$, and $\widetilde{A}[s]$ to denote the $s$-atom in the tuple.

Let $\psi$ be a global formula. We define the notion $\psi \in \widetilde{A}$ as follows: 
\begin{enumerate}
\item for every $s \in Ag$, for every $\alpha \in CL_s$, $\alpha@s \in \widetilde{A}$ iff $\alpha \in \widetilde{A}[s]$; 
\item for every $\Not \psi \in CL$, $\Not \psi \in \widetilde{A}$ iff $\psi \not\in \widetilde{A}$; 
\item for every $\psi_1 \lor \psi_2 \in CL$, $\psi_1 \lor \psi_2 \in \widetilde{A}$ iff $\psi_1 \in \widetilde{A}$ or $\psi_2 \in \widetilde{A}$.
\end{enumerate}

Given atoms $A, A' \in AT_s$, define the {\bf local} relation $\Funnyto_\ell$ as follows: $A \Funnyto_\ell A'$ if and only if

\begin{enumerate}
%condition for $\nxt$
\item for every $\nxt \alpha \in CL_s$, $\nxt \alpha \in A$ iff $\alpha \in A'$.
%condition for $\ominus$
\item for every $\ominus \alpha \in CL_s$, $\ominus \alpha \in A'$ iff $\alpha \in A$.
\end{enumerate}

%$\Funnyto_\ell$ is defined with keeping two issues in mind, first is the reasoning about the $\nxt$ modality in the proof of truth claim, whereas the second is reasoning about $\ominus$-modality.

The communication constraints are defined as follows: consider atoms $A \in AT_s$ and $B \in AT_{s'}$; define the {\bf communication} relation $\Funnyto_\lambda$ as follows. $A \Funnyto_\lambda B$ if and only if

\begin{enumerate} 

\item There exists $A' \in AT_s$ such that $A' \Funnyto_\ell A$; 

\item There exists $B' \in AT_{s'}$ such that $B' \Funnyto_\ell B$;

\item $!a_{s'} \in A'$, and $?a_{s} \in B$.
 
\end{enumerate}

We define local states for agent $s$ as $Q_s=AT_s \times U_s$. 
%For any local state $q =(A,u) \in Qs$, we define $pre(q)=\{q' \in Q_s \mid q' \step{P'}q,~\mbox{for some $P'\subseteq P$}\}$ and $post(q)=\{q' \in Q_s \mid q \step{P'}q',~\mbox{for some $P'\subseteq P$}\}$. The global states $\widetilde{Q}$ are defined as follows: $\widetilde{Q} \subseteq (Q_{s_1}\cdots \times Q_{s_n})$ such that for every $((A_1,u_1),\cdots,(A_n,u_n)) \in \widetilde{Q}$ if for every $1 \le i \le n$, $pre(A_i,u_i)\ne \emptyset$ then $\psi_0 \in (A_1,\cdots,A_n)$. 
%
We use $\widetilde{X}$, $\widetilde{Y}$, to represent members of $\widetilde{Q}$, and $\widetilde{X}(A)[s]$,  $\widetilde{X}(u)[s]$ {\it etc.}, to denote the elements of the tuple in the $s^{th}$ component.

%We say that a global state $\widetilde{X}$ is {\bf self-contained} if for every $a \ne b \in Ag$, for every $\ominus_{s'} \alpha \in CL_s$, $\ominus_{s'} \alpha \in\widetilde{X}(A)[s]$ iff $\alpha \in\widetilde{X}(A)[b]$. 

We are now ready to associate an SCA with the given formula. For $s \in  Ag$, $\Sigma_s \defn 2^{P_s\cup \calM}$ constitute the distributed alphabet over which the SCA is defined. 

\begin{dfn}
Given any formula $\psi_0$, the {\bf SCA associated with} $\psi_0$ is defined by: 

\[ S_{\psi_0} \defn ((Q_{s_1},F_{s_1}), \ldots, (Q_{s_n},F_{s_n}), \to, Init) \] where:
 
\begin{enumerate}

  \item $Q_s = \{(A,u) \mid A\in AT_s, u \subseteq U_s\}$ . 

 % Let global state $\widetilde{Q}$ be a subset of $Q_1\times \cdots \times Q_n$ such that an $n$-tuple  $\langle (A_1, u_1, W_1, R_1) \cdots (A_n, u_n, W_n, R_n)\rangle \in \widetilde{Q}$  if $\langle(A_1, W_1, R_1) \cdots (A_n, W_n, R_n)\rangle \in \mathcal{T}$. We use $\widetilde{X}$, $\widetilde{Y}$, to represent members of $\widetilde{Q}$, and $\widetilde{X}(A)[i]$,  $\widetilde{X}(u)[i]$ {\it etc.}, to denote the elements of the tuple in the $i^{th}$ component.

  %\item $I_s= \{(A,u) \in Q_s \mid \ominus False \in A, \forall !a_{s'},?a_{s'}\in CL_s, \{\lnot !a_{s'},\lnot ?a_{s'}\}\subseteq A, u=\emptyset\}$.
%the definition of initial states makes sure that they have no incoming transitions
  \item $F_s = \{(A,u)\in Q_s \mid\nxt False \in A, u = \emptyset\}$.
%the definition of final states makes sure that they have no outgoing transitions.

  \item $Init = \{((A_1, \emptyset), \ldots, (A_n,\emptyset)) \mid \psi_0 \in 
(A_1, \ldots, A_n)$, and $A_s$ is initial for each $s$ $\}$.

  \item $(A,u) \step{P'}_s (B,v)$, where $A,B \in AT_s$, iff
  
  \begin{enumerate} 

   \item $P' = \{p \in B\cap P_s\} \cup \{a \in \calM \mid ~\mbox{$!a_{s'}$ or $?a_{s'}$ is in $B$ for some $s'\in Ag$}\}$.

%  \item there exists a formula $\nxt \alpha \in subf_s$ such that $\nxt   \alpha \in A$.

   \item $A \Funnyto_\ell B$.

   \item The set $v$ is defined as follows: 

   \[ v = \left\{ \begin{array}{ll}
                \{ \Diamond\alpha \in B \mid \alpha \not \in B \} &                       \mbox{if $u = \emptyset$} \\
                \{ \Diamond\alpha \in u \mid \alpha \not \in B \} &                       \mbox{otherwise} 
                \end{array}
                \right. \]

  \end{enumerate}

  \item $(A,u) \step {\lambda} (B,v)$ iff    $A \Funnyto_\lambda B$.

  \item For every $s_i \ne s_j \in  Ag$, for every $a \in \calM$, for every $(A,u)  \in Q_{s_i}$, if $!a_{s_j} \in A$ then there exists $(B,v)  \in Q_{s_j}$ such that $(A,u) \step {\lambda}   (B,v)$.
\end{enumerate}
\end{dfn}
\noindent We denote $S_{\psi_0}$ by $S_0$ and assert the following with the proof in the appendix.
\begin{thm}
\label{thm:pLTL-sat-thm}
$Models(\psi_0)=\posetlang{S_0}$.
\end{thm} 

The above theorem says that, for a given choreography specification $\psi_0$ in $p$-LTL, the set of service implementations $S_0$, constructed using the above algorithm, actually conform to the behaviour specified by $\psi_0$. This is so as every partial order execution of $S_0$ is actually a model of the formula $\psi_0$ and vice versa. Thus, every choreography expressed in $p$-LTL is realizable.

%---------------------------------------------------------------------------------------------

\section{Implementation} 
The realizability algorithm for choreographies formulated in $p$-LTL, as given in the previous section, has been implemented.  (The program is available from the authors on request.) Now, we briefly explain the program and discuss the experimental results.

This program is written in C and takes formulas as input, in text form. The input formula is preprocessed and converted to a tree form. First, we find the number and names of services from the input. We maintain two different arrays for positive and negative formulas in the closure sets of each service. The sizes of these sets are decided at run time and obtained from the size of input. We read the formula tree and identify subformulas pertaining to the services and put them in the closure set of the respective services. Extra formulas in the closure sets are generated and stored in tree form in the same arrays.

We consider the services one after another and generate atom sets for each and store them in a doubly linked list. A single atom in the NFA of a particular service $s$ is interpreted as a boolean array of length $|CL_s|$ and stored as a number between $0$ and $2^{|CL_s|}-1$. Similarly, the set of unfulfilled $\Diamond$-requirements is taken as a boolean array of length $|U_s|$ and  stored as a number between $0$ and $2^{|U_s|}-1$. The states of $s$ are another doubly linked list where each entry contains two integers, one from the atom set and another from the $\Diamond$-requirement set. Thereafter, we take states from the state list two at a time and check whether the transition properties hold. If they do, the pair is put in the list for transitions, else dropped. This way, we generate transition set for each service $s$ and store them in the respective list. Once, we have obtained all the transition sets (over all the services), we take two transitions from the lists of two different services and check whether they satisfy the properties pertaining to the coupling relation. If they do, we add the pair to the doubly linked list for the coupling relation else we drop it.

In the following table, we present some of the experimental results
obtained from our implementation. The program is run on a laptop with
1.8 GiB RAM and a dual-core 2.10 GHz processor (Intel Pentium B950)
and 32-bit OS (Ubuntu 12.04). We fix the number of services to $2$
and input choreographies with different number of local modalities and
send \& receive propositions.

\begin{center}
\begin{tabular}{l*{5}{c}r}
Size      & Local & Send-Receive & States & Transitions & Couplings  & Time (in ms) \\
\hline
7 	& 2 	& 0 	& 16 	& 32 	& 0 	&   10  \\
9 	& 4 	& 0 	& 64 	& 128  	& 0  	&   14  \\
5 	& 0 	& 2 	& 8 	& 16  	& 16    &   $<1$  \\
5 	& 1 	& 2 	& 80 	& 320  	& 8192  &   520  \\
11	& 2	& 4	& 128	& 512	& 32768	&   1710 \\
\end{tabular}
\end{center}
For different sizes of input choreographies, number of local temporal modalities ($\nxt,\Box,\Diamond$ etc.), and number of send \& receive propositions, we compute the number of states and transitions (over both the local automata), number of couplings ($\lambda$-transitions across the two automata) and the time taken (in ms) to generate the service implementations. Note that, when there are no send \& receive propositions in the input, there are no couplings in the generated SCA. Further, as we increase the number of send \& receive propositions by 2 with an attendant increase in local modalities by 1, the number of coupling shoots up by a multiple of 4. In fact, our suspicion is that too many useless (unreachable) states are being generated and consequently, the number of couplings is on a higher side.

%\section{Related Work}
\section{Discussion}
We have suggested that partial orders are a natural means for talking
of web service interactions and proposed a decidable local logic for specifying 
global conditions on interacting web services such that every formula specifies 
a realizable choreography. 

We have considered choreographies with finite partial order executions. A
natural question relates to choreographies with infinite executions. The
logic can be easily interpreted over such infinite behaviours as well, and
extending SCAs to run on them is straightforward as well. We can thus show
an analogue of Theorem \ref{thm:pLTL-sat-thm}, asserting realizability of 
specifications in the logic over infinite executions, but the technical details 
require some work.

The realizability algorithm which we have given in the paper is quite 
inefficient, in space as well as time, and needs improvement. Similarly
the implementation needs to be fine tuned to make use of partial order
methods and symmetries present in the global configuration space.

Another important area of further research would be improving the quality 
of the solution, that is, to move beyond realizability to realizing 
implementations with desired performance characteristics. For instance,
rather than specifying a fixed number of services {\em a priori}, we
can ask for the minimal number of services realizing a choreography.
The tradeoff between number of services and number of states of each
service or communications between them can be relevant for applications.

A closely related question is when service formulas of $p$-LTL are
used to specify service types instead of individual services. Multiple
instances of services of each type may compose together consistently. 
Such a logic would clearly be more expressive, and its realizability
problem is challenging.

An important theoretical question that arises from the discussion in
the paper is the identification of the largest (satisfiable) subclass 
of LTL choreography properties that are realizable. While the paper
presents a syntactic subclass that is sufficient, we need to expand
it further.

While the paper discusses an initial theoretical investigation, what would
be more interesting is the development of tools that facilitate analysis
of specialized classes of choreographies, and we intend to pursue this.

%\newpage
%\nocite{*}
\bibliographystyle{eptcs}
\bibliography{ref}
\newpage
\appendix
\section{Appendix: Proof of Correctness}
We now prove theorem \ref{thm:pLTL-sat-thm}.
 
\begin{lem}
\label{lem:pLTL-sat-lemma1}
$\posetlang{S_0}\subseteq  Models(\psi_0)$. 
\end{lem}
\begin{proof}
Let $D \in \posetlang{S_0}$.  There is an accepting run $\rho:\calC_D \to 
\widetilde{Q}$ of $S_0$ on $D$.  Let $e$ be an $s_i$-event of $D$. We associate 
an $s_i$-atom $A_e$ with $e$ below. First let $c$ be a configuration with 
$e$ at the $i$-th position. Then, $\rho(c)[i]$ is a tuple $(A,u)$, set 
$A_e = A$ and similarly $u_e=u$.

We can define the valuation function for events of $D$ as follows: for all 
$e \in E_s$, $V(e) \defn (A_e \cap P_s)\cup\{a\in \calM \mid !a_{s_j} \in A_e\}$.

The following assertion can be proved by induction on the structure of formulas 
in $CL_s$.

{\bf Claim:} For all $\alpha \in CL_s$, for all $e \in E_s$, 
$D, e \models_s \alpha$ iff $\alpha \in A_e$. 

Assuming the claim, we show that $D \in Models(\psi_0)$. Note that
by construction of the automaton, if $((A_1,u1) \ldots, (A_n,u_n)) \in Init$ then
$\psi_0 \in (A_1, \ldots, A_n)$. 

Let $c_0 =(\bot_1,\cdots,\bot_n)$ be the initial global configuration.
$\rho(c_0) \in Init$. Thus, we only need to show by above that for
any global formula $\psi$, we have:  $D, c_0 \models \psi$ iff
$\psi \in (A_1, \ldots, A_n)$. This is shown by an easy induction.

\begin{itemize}
\item ($\psi_0=\alpha@s$): $\psi@s \in  (A_1, \ldots, A_n)$ iff 
$\alpha \in A_{e_s}$ by the definition of membership in global atom\\
				iff $D, \bot_s \models_s \alpha$ by Claim above\\
				iff $D, c_0 \models \alpha@s$ by semantics. 
\end{itemize} 
The other cases are similar, by applying the induction hypothesis.
Thus, $D, c_0 \models \psi_0$ and hence $D \models \psi_0$. That is, 
$D \in Models(\psi_0)$, as required.

We proceed to prove the claim.

{\bf Proof:} 

\begin{description}
\item[($\alpha = p$)]$D, e \models_s p$ iff $p \in V(e)$ by definition of local satisfiability\\
						iff $p \in A_e$ by the definition of $V(e)$.

\item[($\alpha = \nxt \beta$)] Suppose $D, e \models_s \nxt \beta$.  We must show that $\nxt\beta \in A_e$.
By the definition of $\models_s$, there exists $e' \in E_s$ such that $e \lessdot_s e'$ and 
$D, e' \models_s \beta$. Let $c \in \calC$ be a configuration with $e$ as the $s$th element. Let $c'\in \calC$ be another configuration with $e'$ as the $s$th element and all the other elements being same as that in $c$. By the definition of run, we have $\rho(c)[s] \step{A_{e'} \cap P_s} \rho(c')[s]$. Therefore, for all $\nxt \beta \in subf_s$,   $\nxt\beta \in A_e$ iff $\beta \in A_{e'}$. $\because$ by the induction hypothesis, $\beta \in  A_{e'}$ hence, we have $\nxt \beta \in A_e$ and we are done.  

Conversely, suppose $\nxt \beta \in A_e$. We must show that 
$D, e \models_s \nxt \beta$. By the induction hypothesis and by the 
semantics of the modality $\nxt$, it suffices to prove that there exists 
$e' \in E_s$ such that $e \lessdot_s e'$ and $\beta \in A_{e'}$. 
Suppose not. Then, $e$ must the $s$-maximal event. So $(A_e,u_e)$ must be a final state, but it is not. Therefore there exists $e' \in E_s$ such that $e \lessdot_s e'$. Let $c \in \calC$ be a configuration with $e$ as the $s$th element. Let $c'\in \calC$ be another configuration with $e'$ as the $s$th element and all the other elements being same as that in $c$. By the definition of run, we have $\rho(c)[s] \step{A_{e'} \cap P_s} \rho(c')[s]$. Therefore, for all $\nxt \beta \in subf_s$,   $\nxt\beta \in A_e$ iff $\beta \in A_{e'}$. Now, $\nxt \beta \in A_e$ as given so $\beta \in A_{e'}$. Thus, we are done.

\item[($\alpha = \Diamond \beta$)] Suppose $D, e \models_s \Diamond \beta$. We must show that $\Diamond \beta \in A_e$. Since $D, e \models_s \Diamond \beta$, there exists $e' \in E_s$ such that $e \le_s e'$, $D, e' \models_s \beta$. That is, by induction hypothesis, there exists $e' \in E_s$ such that $e \le_s e'$, $\beta \in A_{s'}$. We need to show that $\Diamond \beta \in A_e$.

Let $e = e_1\lessdot_s e_2\lessdot_s\cdots \lessdot_s e_k=e'$ be the sequence of events through which $e'$ is reached from $e$.
We show that $\Diamond \beta \in A_e$ by a second induction on 
$l = k-1$.

{\it Base case:} $(l=0)$. 

Then, $k=1$ and so $D, e \models_s \beta$. By the main induction hypothesis, $\beta \in A_e$ and (by the definition of atom), $\Diamond \beta \in A_e$. 

{\it Induction step:} $(l > 0)$.

By the semantics of the modality $\Diamond$, $D, e \models_s \lnot \beta$ and 
$D, e_2 \models_s \Diamond \beta$. Therefore, by the secondary induction 
hypothesis, $\Diamond \beta \in A_{e_2}$. From the definition of $\to$, we 
have $\nxt (\Diamond \beta) \in A_e$.  By the main induction hypothesis, we 
have $\lnot \beta \in A_e$ as well. Combining these facts and using the definition of an atom, we see that $\Diamond \beta \in A_e$ as required. 

Conversely, suppose $\Diamond \beta \in A_e$. We must show that $D, e \models_s 
\Diamond \beta$. Since $\rho$ is an accepting run of $S_0$, there is a maximal
event $e' \in E_s$. Now suppose that $\lnot \beta \in A_{e''}$ for every
$e \leq e'' \leq e'$. Then by an argument similar to the above, we can show
that $\nxt \Diamond \beta \in A_{e''}$ for every $e \leq e'' \leq e'$. Thus
we get $\nxt \Diamond \beta \in A_{e'}$ at the maximal event $e'$ contradicting
the fact that $\nxt \False \in \in A_{e'}$. Thus, there exists $e''$ such that
$e \leq e'' \leq e'$ and $\beta \in A_{e'}$. Then what we need follows by
induction hypothesis.

%\textbf{Proof:} 
%

\item[($\alpha = \ominus \beta$)]
($\To$) Given $D, e \models_s \ominus \beta$. By the definition of local satisfiability, there exists $e' \in E_s$ such that $e' \lessdot_s e$ and $D, e' \models \beta$. Let $c \in \calC$ be a configuration with $e$ as the $s$th element. Let $c'\in \calC$ be another configuration with $e'$ as the $s$th element and all the other elements being same as that in $c$. By the definition of run, we have $\rho(c')[s] \step{A_{e} \cap P_s} \rho(c)[s]$. By the definition of $\to_s$, $\ominus \beta \in A_{e}$ iff $\beta \in A_{e'}$. By induction hypothesis, $\beta \in A_{e'}$, so $\ominus \beta \in A_e$ and we are done. 

($\From$) Given $\ominus \beta \in A_e$. It suffices to show that there exists $e' \in E_s$ such that $e' \lessdot_s e$ and $\beta \in A_{e'}$.
Suppose there is no $e'$ such that $e' \lessdot_c e$. That is, $e$ is the $s$-minimum event. Then, $(A_e,u_e)\in I_s$. So, for every $\ominus \gamma \in CL_s$, $\ominus \gamma \not \in A_e$. This is a contradiction. Therefore, there exists $e' \in E_s$ such that $e' \lessdot_s e$.
Now, let $c \in \calC$ be a configuration with $e$ as the $s$th element. Also, let $c'\in \calC$ be another configuration with $e'$ as the $s$th element and all the other elements being same as that in $c$. By the definition of run, we have $\rho(c')[s] \step{A_{e} \cap P_s} \rho(c)[s]$. By the definition of $\to_s$, $\ominus \beta \in A_{e}$ iff $\beta \in A_{e'}$. Therefore, $\beta \in A_{e'}$ as we already have $\ominus \beta \in A_e$. Thus, we have shown that there exists $e' \in E_s$ such that $e' \lessdot_s e$ and $\beta \in A_{e'}$ and we are done.

\item[($\alpha = !a_{s_j}$)]
($\To$) Given $D, e \models_s !a_{s_j}$. There exists $e' \in E_{s_j}$ such that $e <_c e'$ and $a \in V(e)$. By the definition of run, $(A_{e},u_{e}) \step{\lambda} (A_{e'},u_{e'})$. By the definition of $\Funnyto_\lambda$, $!a_{s_j} \in A_e$.

($\From$) Given $!a_{s_j} \in A_e$. There exists $(A_{e'},u_{e'}) \in Q_{s_j}$ such that $(A_{e},u_{e}) \step{\lambda} (A_{e'},u_{e'})$. By the definition of the run $e <_c e'$. By the definition of $V$, $a \in V(e)$. Therefore, $D, e \models_s !a_{s_j}$.

\item[($\alpha = ?a_{s_j}$)] The reasoning about $?a_{s_j}$ is similar to that of $!a_{s_j}$ as given above.
%\end{description}
\end{description}
\end{proof}

\begin{lem}
\label{lem:pLTL-sat-lemma2}
$Models(\psi_0)\subseteq \posetlang{S_0}$. 
\end{lem}

\begin{proof}
Conversely, suppose $D \models \psi_0$, where $D = (E_{s_1}, \ldots, E_{s_n},\le_{s_1},\cdots,\le_{s_n},<_c V)$. To show that $D$ is a member of $\posetlang{S_0}$, we have to construct an accepting run of $\posetlang{S_0}$ on $D$.
%----------------------------------------------------------------------------------------

%------------------------------------------------------------------------------------------
\noindent For every $s \in Ag$, for every $e \in E_s$, define the set $A_e$ as follows:
\[A_e=\{\alpha \in CL_s \mid D,e \models \alpha\}\]
%
%Verify that $A_e$ is a valid $s$-atom, satisfying all the conditions enumerated above.
%
Let $e_s^0$ be the minimum event in $E_s$. We construct $A_{\bot_s}$ from $A_{e_s^0}$ as follows:
\[A_{\bot_s}=\Delta A_{\bot_s} \cup \{\lnot \alpha \in CL_s \mid \alpha \not \in \Delta A_{\bot_s}\} \cup \{\alpha \lor \beta \in CL_s \mid \alpha \in \Delta A_{\bot_s}\} \cup \{\nxt \Diamond \alpha \mid \Diamond \alpha, \lnot \alpha \in \Delta A_{\bot_s}\}~~~\mbox{where}~~\]
\[\Delta A_{\bot_s}= \{\lnot p \mid p \in CL_s \cap P_s\} \cup \{\lnot !a_{s'},\lnot ?a_{s'}\mid !a_{s'},?a_{s'}\in CL_s\}\cup \{\lnot \ominus \beta \mid \ominus \beta \in CL_s\}\cup \delta A_{e_s^0}~~~\mbox{and}~~\]
\[\delta A_{e_s^0}=\{\nxt \alpha \in CL_s \mid \alpha \in A_{e_s^0}\}\cup \{\alpha \mid \ominus \alpha \in A_{e_s^0}\} \cup \{\ominus False\}\]
For every $s \in Ag$, for every $f \in E_s'$, define the set $u_f$ inductively as follows:
\begin{itemize}
\item $u_{\bot_s}=\emptyset$,
\item for every $f,f'\in E_s'$ such that $f\lessdot_s f'$,
   \[ u_{f'} = \left\{ \begin{array}{ll}
                \{\Diamond \alpha  \in A_{f'} \mid \alpha \not \in A_{f'} \} &                       \mbox{if $u_e = \emptyset$} \\
                \{ \Diamond \alpha \in u_f \mid \alpha \not \in A_{f'} \} &                       \mbox{otherwise} 
                \end{array}
                \right. \]

\end{itemize}
Now, for any configuration $c=(f_{s_1},\ldots,f_{s_n})$ in $\calC_D$ define
\[\rho(c)=\langle(A_{f_{s_1}},u_{f_{s_1}}),\ldots,(A_{f_{s_n}},u_{f_{s_n}})\rangle.\]
%
%Verify that $\rho:\calC_D \to \widetilde{Q}$ as defined above is a valid run of $S_0$ on $D$, and it is accepting. 
%
It is now easily shown that $\rho$ is an accepting run of $\posetlang{S_0}$ on $D$ and hence, $D \in \posetlang{S_0}$ and we are done.
\end{proof}

The two foregoing lemmae \ref{lem:pLTL-sat-lemma1} and \ref{lem:pLTL-sat-lemma2}, 
together, give us the theorem \ref{thm:pLTL-sat-thm}.

%\section{Prefaces}

%Volume editors may create prefaces using this very template,
%with {\tt $\backslash$title$\{$Preface$\}$} and {\tt $\backslash$author$\{\}$}.

\end{document}